\documentclass[aps,pra,floatfix,amsmath,superscriptaddress,twocolumn]{revtex4-1}
\usepackage{amssymb}
\usepackage{amsmath}
\usepackage{amsthm}
\usepackage{hyperref}
\usepackage{bbm}

\newcommand{\Eqref}[1]{Eq.~(\ref{#1})}                                
\newcommand{\Eqsref}[1]{Eqs.~(\ref{#1})}
\newcommand{\Secref}[1]{Section~\ref{#1}}
\newcommand{\Aref}[1]{Appendix~(\ref{#1})}
\def\id{\ensuremath{\mathbbm{1}}} 
\newcommand{\csch}{\mathrm{csch}} 
\newcommand{\mr}[1]{\mathrm{#1}}
\newcommand{\RR}{\ensuremath{\mathbbm{R}}} 
\newcommand{\NN}{\ensuremath{\mathbbm{N}}} 
\def\cH{{\cal H}}
\def\LZR{\ensuremath{\mathrm{L}^{\!\!\rule[-0.5ex]{0mm}{0mm}2}(\RR)}}
\renewcommand{\Re}{\mathrm{Re}}

\def\sz{\sigma_z}

\newcommand{\sh}{\mr{sh}}
\newcommand{\ch}{\mr{ch}}

\def\diag{\mr{diag}}
\newcommand{\bea}{\begin{eqnarray}}
\newcommand{\eea}{\end{eqnarray}}
\def\bi{\begin{itemize}}
\def\ei{\end{itemize}}

\newcommand{\R}{\RR}
 
\newcommand{\one}{\id}
\def\tr{\mathrm{tr}}
\def\ket#1{\left| #1\right>}

\newtheorem{theorem}{Theorem}

\newtheorem{lemma}[theorem]{Lemma}

\renewcommand{\url}[1]{\relax} 
\begin{document}

\author{G. Giedke}
\affiliation{Max--Planck Institute for Quantum Optics, Garching, Germany}
\author{B. Kraus}
\affiliation{Institute for Theoretical Physics, University of
Innsbruck, Innsbruck, Austria}
\title{Gaussian Local Unitary Equivalence of $n$-mode Gaussian States\newline  
and Gaussian LOCC transformations}

\date{$3^\mathrm{rd}$ November, 2013}

\begin{abstract}
  We derive necessary and sufficient conditions for arbitrary multi--mode
  (pure or mixed) Gaussian states to be equivalent under Gaussian local
  unitary operations. To do so, we introduce a standard form for Gaussian
  states, which has the properties that (i) every state can be transformed
  into its standard form via Gaussian local unitaries and (ii) it is unique
  and (iii) it can be easily computed. Thus, two states are equivalent under
  Gaussian local unitaries iff their standard form coincides. We explicitly
  derive the standard form for two-- and three--mode Gaussian pure states.\\
  We then investigate transformations between these classes by means of
  Gaussian local operations assisted by classical communication. For
  three-mode pure states, we identify a global property that cannot be created
  but only destroyed by local operations. This implies that the highly
  entangled family of symmetric three--mode Gaussian states 
is not sufficient to generated all three-mode
  Gaussian states by local Gaussian operations.
\end{abstract}

\maketitle

\section{Introduction}

Since most applications of quantum information rest upon the subtle properties
of multipartite quantum systems, the qualification and quantification of
multipartite entanglement is a central task of quantum information
theory. Whereas the bipartite case is for finite as well as for certain
infinite dimensional systems well understood, many questions are still open in
the multipartite setting \cite{HHHH09}.

The set of \emph{Gaussian states} still plays a major role in current
experiments dealing with continuous quantum variables, as it comprises those
states that are processed in most experiments. This, and the mathematical
simplicity of those states, which can be fully characterized by the finite set
of first and second moments, are the reason why mainly Gaussian states have
been investigated in the context of continuous-variable (CV) quantum information
\cite{Weedbrook2012}.

Regarding the entanglement properties of Gaussian states, it has been shown
that, as in the finite dimensional case a separable state has positive partial
transpose and that there exist entangled states with positive partial
transpose \cite{WeWo01}. However, for party $A$ possessing 1 mode and party
$B$ arbitrary many, it has been shown that partial transposition leads to a
necessary and sufficient condition for separability \cite{WeWo01}. For the
general bipartite case, i.e., when both parties possess an arbitrary number of
modes efficiently testable necessary and sufficient conditions of separability
have been derived \cite{GKLC01b,HyEi06}. In contrast to the case of finite
dimensional Hilbert spaces, the question of which states can be distilled to
pure entanglement has been solved for bipartite Gaussian states. In fact, is
was shown \cite{GDCZ01} that a bipartite Gaussian state is distillable iff its
partial transpose is not positive semidefinite 
\footnote{Recall that in the bipartite finite dimensional case this is only 
known to be 
a necessary condition for a state to be distillable. There, it is
long-standing (though still open) conjecture, that undistillable states which
do not have a positive semidefinite partial transpose \cite{DCLB99,DSST99}
exist.}. 
In Refs.~\cite{DVV79,ESP02,Fiu02,GiCi02} the problem of manipulation of
Gaussian states has been studied. In particular, in \cite{GiCi02,Fiu02} the
most general operations transforming Gaussian states to Gaussian states were
studied. These operations are called Gaussian operations. In \cite{GiCi02} it
has been proven that it is not possible to distill Gaussian states using
Gaussian operations (see also \cite{ESP02,Fiu02}).

The knowledge about entanglement in the multipartite setting is still far from
complete, although a large number of (mostly) partial results has been
obtained. The generation of pure multipartite entangled Gaussian states was
discussed in \cite{LoBr00}. A classification of multipartite entanglement
classes of arbitrary three-mode Gaussian states have been presented in
\cite{GKLC01}. Practical criteria for the certification of genuine
multipartite entanglement were derived in \cite{LF02}. A general solution to
the multipartite separability problem in the Gaussian case was provided by the
Gaussian entanglement witnesses and related semi-definite programs studied in
\cite{HyEi06}. A large number of quantum optical experiments demonstrating
multi-mode entanglement in increasingly large systems
\cite{ATY+03,YAF04,ST+Peng07,YU+Furusawa08,MERR08,PM+Pfister11,SZ+Peng12},
culminating in $10\,000$-mode (time-bin) entanglement reported in
\cite{YU+Furusawa13}. Moreover, several standard entanglement measures have
been adapted to the Gaussian setting (as, e.g., robustness (\cite{GiCi02},
obtainable from a semi-definite program as described in \cite{HyEi06}) or
Gaussian localizable entanglement \cite{FiMi07}) and notions such as GHZ-like
states \cite{ASI06}, maximal entanglement (as quantified by bipartite
entanglement) \cite{FFL+09}, monogamy of entanglement \cite{AdIl06} have been
specialized to the Gaussian setting.

Despite these advances, the study of multipartite entanglement is still in an
early stage. One method to gain more insight into the entanglement properties
of multipartite states is to investigate their interconvertibility. An
important fine-grained classification of multipartite entangled states sorts
them according to convertibility by local unitaries, leading to the notion of
local unitary (LU) equivalence \cite{GRB98,Sud00,BaLi01,Kraus2010}. Clearly,
two LU equivalent states possess the same amount of entanglement and are
equivalent as a non-local resource. LU equivalence leads to a very detailed
classification of multipartite states with a continuum of inequivalent
classes. A more coarse-grained (and therefore often more insightful) picture
emerges if a larger class of transformations is allowed. Especially useful for
entanglement classification is to allow for non-trace-preserving operations
[(partial) measurements] and classical communication between parties which
leads to the set of stochastic local operations and classical communication
(SLOCC) \cite{BPR+99}. SLOCC play an important role in entanglement theory
\cite{Nie99,Vid99b,DVC00,VDMV02,OsSi12,VSK13,GoWa13}. SLOCC-convertibility
gives rise to fewer equivalence classes than LU-equivalence and in some cases
only finitely many \cite{Nie99,Vid99b,DVC00} SLOCC-classes exist.

For Gaussian states, it is reasonable to consider convertibility under
Gaussian operations. Conversion (of mixed states) under trace-preserving local
Gaussian operations (LOG; not necessarily unitary) was investigated for the
two-mode case in \cite{EP02} and for the general bi- and tripartite setting in
\cite{WLY+03,WLY04}, while transformation under trace-nonpreserving local
Gaussian operations has been investigated in \cite{GECP03} for pure bipartite
states. The equivalence of Gaussian states under Gaussian local unitaries
(GLU) was studied for the (mixed) bipartite setting in \cite{Wolf08} and for
more parties in \cite{WLY04,AdIl06,Adesso06,Serafini2007}. In
\cite{Serafini2007} and \cite{Adesso06} standard forms for ``generic''
$n$-mode mixed and pure states were introduced. The case of pure three-mode
states has been studied in detail in \cite{ASI06}. There, it is shown (for
``generic'' pure Gaussian states) that the GLU equivalence classes are
characterized by three positive numbers (related to local purities) and a
simple standard form was derived.

The aim of this paper is to derive a standard form for arbitrary Gaussian
states which has the properties that (i) every state can be transformed into
its standard form via Gaussian local unitaries, (ii) it is unique, and (iii)
it can be easily computed. Due to these properties, the solution to the
Gaussian LU--equivalence problem follows easily. We then focus on pure
Gaussian three--mode states and show that any such state is characterized by
the three local purities. The standard form of those states is used to
investigate the manipulation of those states using GLOCC. We show that the
completely symmetric states, which are sometimes referred to maximally
entangled states, \emph{cannot} be used to obtain an arbitrary state via
GLOCC.

The remainder of the paper is organized as follows. In Sec.~\ref{Sec_Intro} we
briefly review the basic concepts and results on Gaussian states needed
later. In Sec.~\ref{Sec_GLU} we present a standard form for arbitrary (pure
and mixed) $n$--mode Gaussian states, where all modes are spatially separated,
and derive the necessary and sufficient conditions for Gaussian state to be
equivalent under Gaussian local unitaries (GLUs). As we will show, this
criterion can efficiently be applied, since it only involves the computation
of the singular value decomposition of $2\times 2$ matrices, independently of
the system size. We will then demonstrate our methods by considering first the
simplest case of two modes, where we show that our standard form coincides
with the one presented in \cite{DGCZ99,Sim00}. In Sec.~\ref{SubsecThreeM}, we
investigate the different GLU--equivalence classes of three--mode Gaussian
states. We show that \emph{any} pure three--mode Gaussian state is
GLU--equivalent to a state with no correlations between the $X$ and $P$
quadratures and that an arbitrary three-mode pure Gaussian state is (up to
GLU) uniquely characterized via the three local purities, i.e., by the
bipartite entanglement between each single mode and the remaining two
modes. This reproduces the results of \cite{ASI06} but shows that they apply
to all three-mode states (not only a subset of generic states).  In order to
obtain more insight into the entanglement properties of three--mode states, we
consider in Sec.~\ref{SecGLOCC} the more general set of Gaussian local
operations assisted by classical communication (GLOCC). In particular, we show
that it is not possible to obtain from the symmetric Gaussian pure three-mode
states (which are sometimes referred to as maximally entangled states or
continuous-variables analogs of both GHZ and W states (``CV GHZ/W-states'',
for short) \cite{LoBr00,ASI06,AdIl06}), all pure
three-mode state via GLOCC. This implies that those states are not, as the
two--mode squeezed states are in the bipartite case, sufficient to obtain
deterministically any other state via local Gaussian operations (and thus not
a Gaussian analog of the maximally entangled set introduced in
\cite{VSK13}). In contrast, we finally present a class of states from which,
in particular, all symmetric states can be obtained via GLOCC. Hence, this
class of states might be called more entangled than the symmetric one.

\section{Preliminaries}
\label{Sec_Intro}
We summarize here some results concerning Gaussian
states and introduce our notation. We consider systems composed of $n$ modes,
i.e., $n$ distinguishable infinite
dimensional subsystems, each with Hilbert space $\cH = \LZR$. To each mode
$k=1,\dots,n$ belong two
canonical observables $X_k, P_k$ which obey the commutation
relation $[X_k,P_k]=i$. Defining $R_{2k-1}=X_k, R_{2k}=P_k$, these relations
are summarized as $[R_l,R_m]=-iJ_{lm}$,
using the antisymmetric $2n\times 2n$ matrix
\begin{equation}
\label{auxdefJn} J\equiv \oplus_{k=1}^n J_1,\quad J_1 \equiv
\left(\begin{array}{cc} 0 & -1 \\ 1 & 0
\end{array}\right),
\end{equation}
where here, and in the following $\oplus$ denotes the direct sum. Let us
denote the unitary displacement operator by
\begin{equation}
\label{displacement}
D(x)=e^{i\sum_{k} (q_k X_k + p_k P_k)}\equiv e^{ix\cdot R},
\end{equation}
where $x =(q_1,p_1,\ldots q_n,p_n)\in\RR^{2n}$. Using this notation, the
characteristic function of a state $\rho$ is defined as
\begin{equation}
\label{defcharc} \chi_{\rho}(x) = {\rm tr} [\rho D(x)].
\end{equation}

Gaussian states are those states for which $\chi$ is a Gaussian multivariant
function of the phase space coordinates, $x$ \cite{MaVe68}, i.e.,
\begin{equation}\label{charfct}
\chi_{\rho}(x) = e^{-\frac{1}{4}x^T \gamma x - id^Tx}.
\end{equation}
Here, $\gamma$ is a real, symmetric, strictly positive $2n\times 2n$ matrix, the
\emph{covariance matrix} (CM), and
$d\in\RR^{2n}$ is a real vector, the
\emph{displacement}. A Gaussian state is
completely determined by $\gamma$ and $d$. Note that both
$\gamma$ and $d$ are directly measurable quantities, as their
elements $\gamma_{kl}$ and $d_k$ are determined by the expectation
values and variances of the operators $R_k$, via
\bea%
\label{defdisplacement}
d_k=\tr(\rho R_k),
\eea%
\begin{equation}%
\label{auxCM1} \gamma_{kl} =  2 \Re\{\tr[\rho (R_k-d_k)
(R_l-d_l)]\}.
\end{equation}%
The displacement of a (known) state can always be adjusted to $d=0$ by a
sequence of local unitary operators applied to individual modes \footnote{The
  unitaries are generated by linear Hamiltonians.}. Thus, the first
moments are irrelevant for both the study of GLU--equivalence classes and the
entanglement contained in the state and will therefore be set to zero.

Not all real, symmetric, positive matrices $\gamma$ correspond to the CM of a
physical state, they also have to satisfy the uncertainty principle. There are
several equivalent ways to characterize valid CMs, which are all useful in the
following. Before we summarize them in Lemma \ref{lemmaCM} let us recall that
a (real) linear transformation $S$ on phase space is called \emph{symplectic}
if it preserves $J$, i.e., if $SJS^T= J$ holds. The group of real symplectic
$2n\times 2n$ matrices is denoted by $Sp_{2n}(\R)$. Let us now state the
conditions for a matrix to be a valid CM.

\begin{lemma} (Covariance Matrices)\\
\label{lemmaCM} A real, symmetric and positive $2n\times2n$ matrix,
$\gamma$, is the CM of a physical state iff one of the following equivalent conditions holds
\begin{subequations}\label{physicalconds}
\begin{equation}\label{ii}
\gamma+J\gamma^{-1}J \geq 0,
\end{equation}
\begin{equation}\label{iii}
\gamma-i J \geq 0,
\end{equation}
\begin{equation}\label{iv}
\gamma = S^T(D \oplus D) S,
\end{equation}
for $S$ symplectic and $D\geq\id$
diagonal.\\
The CM $\gamma$ describes a \emph{pure} state iff equality holds in
  \Eqref{ii} or, equivalently, iff $D=\id$ in \Eqref{iv}, i.e., iff
  $\det\gamma=1$.
\end{subequations}
\end{lemma}

The proofs of these statements can be found in \cite{MaVe68,WeWo01,Fol89} respectively. As an example of a valid CM, let us recall that the CM of an arbitrary pure two-mode states ($1\times 1$ case), $\gamma$, can
be written as \cite{SMD94}
\begin{equation}\label{auxgamma2}
 \gamma = (S_1\oplus S_2)\left( \begin{array}{cc} \cosh r \one& \sinh r\sigma_z\\
 \sinh r \sigma_z&\cosh r\one
\end{array} \right)(S_1^T\oplus S_2^T).
\end{equation}
Here and in the following
$S_{1,2}$ are local symplectic matrices,
$r\geq0$, and $\sigma_x,\sigma_y,\sigma_z$ denote the Pauli operators. The parameter $r$ contains all information about
the entanglement of the state, whereas $S_1$ and $S_2$ contain
information about local squeezing \footnote{Given a CM $\gamma$ in its
block form (\ref{defBlockmgamma}), one can readily find its pure
state standard form using the following procedure:
We have $S_k=O_kD_kO_k'$, where $O_k,O_k'$ are rotations and
$D_k=\diag(e^{r_k},e^{-r_k})$.
The six matrices are determined as follows: $O_{1(2)}$ diagonalize $A
(B)$. The rotations $O_k'$ realize the singular value
decomposition of $D_1^{-1}O_1^TCO_2D_2^{-1}$. The two-mode
squeezing parameter $r$ is given by $\cosh r=\sqrt{\det(A)}$,
while the squeezing parameters $r_1,r_2$ of $S_k$ can be
calculated by the trace of $A$ and $B$, resp.: $\cosh 2 r_1 = (\tr
A)/(2\cosh r), \cosh 2 r_2 = (\tr B)/(2\cosh r)$.}. An example of a pure state would be the two--mode squeezed state, whose CM is given by \Eqref{auxgamma2} with $S_1=S_2=\one$.

Whenever we consider a bipartite splitting of the state ($n$ modes at one side
and $m$ modes at the other, which we call $n\times
m$ case in the following) we might write the CM in the index--free block form
\begin{equation}
\label{defBlockmgamma}
 \gamma = \left( \begin{array}{cc} A&C\\ C^T&B
\end{array} \right).
\end{equation}
Here $A$, $B$ and $C$ are $2n\times 2n$, $2m\times 2m$, and $2n\times 2m$
matrices, respectively. Note that $A$ ($B$) is the CM corresponding to the
reduced state of the first (second) system, respectively. The correlations
between both systems are described by the matrix $C$, which vanishes for
product states.

Since we are interested in Gaussian local unitary equivalence classes in this
paper, we also review here how the CM $\gamma$ (and the displacement $d$) of a
Gaussian state $\rho$ change under the evolution of a Gaussian unitary
operator $U$. As can be easily verified, a unitary operator 
transforms any Gaussian state into a Gaussian state (i.e., describes a
Gaussian operation) iff there exists a symplectic matrix $S$ and a real vector
$r\in\RR^{2n}$, such that
$U^\dagger R U=S R+r$. Discarding the irrelevant displacement, the CM 
transforms according to \footnote{One could also use the Heisenberg equation,
  i.e., $dA/dt=i[H,A]$ for the mode operators ($X_k,P_k$) to obtain this
  result.}
\begin{align}
\label{time-evCM}
\gamma^\prime&=S \gamma S^T. 
\end{align}

The most general $S\in Sp_{2n}(\RR)$ can be written as $S=O_1 DO_2$, where
$O_{1,2}$ are real orthogonal and symplectic matrices and
$D=\mbox{diag}(r_1,\dots,r_n,1/r_1,\dots,1/r_n)$, with $r_i\in\R^+$
\cite{ADMS95b}; for $D=\one$, $S$ is called a passive operation, otherwise it
is called active. Apart from describing Gaussian unitary operations,
symplectic matrices can also be used to derive a simple normal form
(Williamson normal form) for arbitrary CM, see \Eqref{iv}. The eigenvalues
$d_i$ of $D$ are called the symplectic eigenvalues of $\gamma$ and are $\geq
1$. They are related to the purity of the corresponding Gaussian state,
$\rho$, since $\tr(\rho^2)$ is given by \cite{Scu98} \bea %
\tr(\rho^2)=|\gamma|^{-1/2}=\prod_{i=1}^n d_i^{-1}, \eea %
where here and in the following, $|\cdot|$ denotes the determinant. This can
be easily verified by noting that $|\gamma|=|\gamma J|= |S^{-1} \oplus_{i=1}^n
d_i \one (S^T)^{-1} J|=| \oplus_{i=1}^n d_i \one |=\prod d_i^2$. The purity
can be utilized to quantify the entanglement contained in pure states. For
instance, the quantity \bea\label{defpurity}%
P(\ket{\Psi})=\tr(\rho_{red}^2)^{-2}, \eea%
where $\rho_{red}$ denotes the reduced density operator of either system $A$
or $B$ of the pure state $\ket{\Psi}$, increases the more entangled
$\ket{\Psi}$ is.  Using the block form of the CM, $\gamma$ [see
Eq.~(\ref{defBlockmgamma})], $P(\ket{\Psi})$ is given by $|A|=|B|$.
\section{GLU--Equivalence and standard form}\label{Sec_GLU}

We consider an arbitrary $n$--mode Gaussian state (pure or mixed), with CM
$\gamma$ and assume a partition of one mode per site. We first derive a
standard form of $\gamma$, $S(\gamma)$, which we show to be unique, easily
computable and to which each CM can be mapped via GLU. Two states are called
GLU-equivalent if their density matrices can be transformed into each other by
Gaussian local unitaries. Thus two Gaussian states with CM $\gamma$
resp. $\Gamma$ are GLU-equivalent iff their CMs can be transformed into each
other by a local symplectic transformation. Due to the fact that the standard
form, which we introduce here, is unique it easily follows that two Gaussian
states are GLU--equivalent iff their standard forms coincides.

We denote in the following by $\gamma_{jk}$ the $2\times2$ matrix describing
the covariances between mode $j$ and $k$. As mentioned before any $2\times 2$
real symplectic matrix can be written as $O_1 \mbox{diag}(r, 1/r) O_2$, with
$r\in \R$ and $O_i$ real orthogonal.  The standard form is reached in two
steps. First, we apply to each mode $j$ the active GLU that symplectically
diagonalizes $\gamma_{jj}$, i.e., $S(\gamma)_{jj}=\lambda_j\one$. This leaves
still the freedom to apply local \emph{passive} operations $S^p_j$ to each
mode $j$, which are given by $O_j\in SO(2)$. In the second step, we fix the
$O_j=\exp(i\alpha_j\sigma_y)$ by considering the off-diagonal blocks
$\gamma_{jk}, j<k$ in turn (row by row, from left to right). First consider
$\gamma_{12}$ and determine its singular values; if they both are zero,
continue with the next block; if they are non-zero but degenerate, then
$\gamma_{12}$, obeying $\gamma_{12} \gamma_{12}^T\propto \one$ and
$|\gamma_{12}|>0$, is proportional to a real special orthogonal matrix, $O$
which we write without loss of generality as $O=e^{i\alpha\sigma_y}\in
SO(2)$. We fix $\alpha_2=\alpha+\alpha_1$ (with $\alpha_1$ being determined
subsequently); if they are non-degenerate and add to zero, then
$\gamma_{12}\propto \sigma_z O$, with $O=e^{i\alpha\sigma_y}\in SO(2)$ and we
fix $\alpha_2=\alpha-\alpha_1$ (we refer to the two cases that $\gamma_{ij}$
is orthogonal as ``degenerate''); otherwise, we fix both $\alpha_1,\alpha_2$
such that $O_1,O_2$ are the unique matrices $\in SO(2)$ such that
$O_1\gamma_{12}O_2^T=\mbox{diag}(d_{12},d_{12}')$, with $d_{12}\geq |d_{12}'|$
\footnote{$O_1$ diagonalizes $\gamma_{12}\gamma_{12}^T$ and $O_2$ diagonalizes
  $\gamma_{12}^T\gamma_{12}$.}. In all four cases $S(\gamma)_{12}$ is
diagonal. Now treat $\gamma_{13}$ (and then all subsequent $\gamma_{jk}$) in
the same manner. If $\alpha_j$ has already been determined in a previous step,
then for non-degenerate singular values of $\gamma_{jk}$ we fix $\alpha_k$ by
diagonalizing $\gamma_{jk}^T\gamma_{jk}$. In this manner, all $\alpha_j$ will
be uniquely determined except in the case that (for some $j$) all
$\gamma_{jk}$ are zero (in which case the mode $j$ factorizes and we set
$\alpha_j=0$) or that for each $j$ there is exactly one non--vanishing
degenerate $\gamma_{jk}$ (in this case we set the undetermined
$\alpha_{j}=0$). Any $n$-mode CM is transformed to its standard form
$S(\gamma)$ by applying the $n$ local active and $n$ local passive unitaries
as described above and we have
\begin{theorem}[Criterion for GLU-Equivalence]\label{th:GLUequiv}
  Any CM $\gamma$ can be transformed into its standard form, $S(\gamma)$,
  by Gaussian local unitaries. Two CMs $\gamma$ and $\Gamma$ are
  GLU-equivalent if and only if $S(\gamma)=S(\Gamma)$.
\end{theorem}

Note that this criterion for GLU--equivalence is valid for both, mixed and
pure states. Let us mention here that an essentially identical form for
$n$-partite $n$-mode Gaussian states was introduced in \cite{Serafini2007} and
that the $n\times n\times n$ case was discussed in \cite{WLY04}. However, the
question whether this is a unique standard form (which is essential for
Theorem~\ref{th:GLUequiv}) was only discussed for generic states.  
Let us close this discussion with a remark on the relation of LU-- and
GLU--equivalence before using
the GLU--criterion to derive the different GLU--classes of 2--mode and
3--mode states.

When studying GLU-equivalence, we restrict the allowed operations to a very
small subset of all local unitaries. Hence, in general, two LU-equivalent
states are not GLU-equivalent. However, for pure Gaussian states in a number
of relevant cases the two notions coincide.  Note that, in particular, if two
pure states are LU--equivalent, then the Schmidt coefficients of these states
across any bipartition must be the same. If we can show that the GLU-classes
of Gaussian states are uniquely characterized by their Schmidt coefficients
across all bipartitions, then it follows that for those Gaussian states
LU--equivalence implies GLU--equivalence. This is actually the case for pure
bipartite Gaussian states, as implied by the results of \cite{GECP03}:.  every
pure $n\times m$ Gaussian state $\ket{\psi}$ is GLU--equivalent to
$\mr{min}\{n,m\}$ two-mode squeezed states $\ket{\psi_\mr{tms}(r_j)}$ with
squeezing parameters $r_1\geq r_2\geq\dots r_{\mathrm{min}\{n,m\}}\geq0$,
which fixes the Schmidt coefficients
$\lambda_{l_1,\dots,l_n}=\prod_{j=1}^n\frac{\tanh^{2l_j}r_j}{\cosh r_j}$ where
$l_j\in\NN$.  Thus, if two pure bipartite Gaussian states are LU--equivalent
they have the same standard form $\otimes_{j=1}^n\ket{\psi_\mr{tms}(r_j)}$ and
therefore are also GLU--equivalent. As we show in Subsec.~\ref{SubsecThreeM}
below, the same implication also holds for pure $1\times1\times1$ Gaussian
states.

\subsection{$1\times 1$ case}\label{subsec:1x1}
Let us first consider the simplest case of two mode Gaussian states. First we apply active transformations to map the reduced
states, $\gamma_{ii}$ to thermal states, $\lambda_i \one$. Since the state is
pure the reduced states must be identical,
i.e., $\lambda_1=\lambda_2=\lambda$. According to the algorithm above we apply
next the orthogonal matrices, $O_1,O_2$ such that $O_1\gamma_{12}O_2^T=D$,
where $D$ is diagonal. Thus, the standard form, $S(\gamma)$ is
\begin{align}
S(\gamma) &=\left(\begin{array}{cc}
\lambda \one &D \\
D& \lambda \one\end{array}\right).
\end{align}
Next, we show that the standard form introduced here coincides in the case of pure
two--mode states with the form \cite{GECP03} \bea \label{StFrom2Modes} \gamma
=\left(\begin{array}{cc}
    \cosh(r)\one &\sinh(r) \sigma_z \\
    \sinh(r) \sigma_z & \cosh(r)\one\end{array}\right),\eea with the squeezing
parameter $r$. Note that due to condition $\gamma J \gamma \geq J$, we have $\lambda\geq 1$. Imposing now the condition that $\gamma$ corresponds to a pure state,
i.e., $\gamma J\gamma=J$ we find $\lambda^2\one + \tilde{D} D=\one$ and
$\lambda \{J,D\}=0$, where here and in the following $\{A,B\}=AB+BA$ denotes
the anticommutator between any operators $A$ and $B$ and $\tilde{D}=\sigma_x
D\sigma_x$. Since $\lambda \geq 1$ must be fulfilled by any CM, it must hold that
$\{J,D\}=J(D+\tilde{D})=0$, which implies that $D=\bar{\lambda} \sigma_z$, for some real $\bar{\lambda}$. Due to
the first condition we get then $\lambda^2-\bar{\lambda}^2=1$, which implies
that we can choose $\bar{\lambda}=\sinh(r)$ and $\lambda=\cosh(r)$, for some
$r\in \R$. Thus the standard form coincides with Eq.~(\ref{StFrom2Modes}).

\subsection{$1\times 1\times 1$ case}
\label{SubsecThreeM}
In this section we identify the different GLU--classes of 3--mode Gaussian
states. First we explicitly provide the general standard form of
Theorem~\ref{th:GLUequiv} for the three--mode case. Then we show that it
considerably simplifies for pure states and prove an exhaustive
parameterization of the pure three-mode states.

\subsubsection{Standard form: Mixed states}
In this section we derive the standard form for an arbitrary $1\times 1\times
1$ Gaussian state. It is convenient to introduce (index-free)
notation for the nine $2\times2$ blocks of $\gamma$ by defining the
matrix $\gamma$ as
\begin{equation}
  \label{eq:1}
\gamma \equiv \left( \begin{array}{ccc}
A&K&L\\
K^T&B&M\\
L^T&M^T&C
\end{array} \right).
\end{equation}
The basis chosen here will be called mode-ordered, as indices referring to the
same mode ($A,B$ or $C$) are grouped. Sometimes the \emph{quadrature-ordered}
basis is used. This is a permutation in which first all the indices referring
to $X$-quadratures appear, followed by those referring to $P$.

As before, we first choose the active transformations to map the reduced
states into thermal states. Using the same notation as before, the real
orthogonal matrices $O_i$, for $i=1,2,3$ are then used to map the
off--diagonal matrices into diagonal matrices. In case the singular values of
all off--diagonal blocks are non--degenerate, we use $O_1$ and $O_2$ to map
$K$ into a diagonal matrix with sorted entries in the diagonal, i.e., $O_1 K
O_2^T \equiv \mbox{diag}(d_{12}^+,d_{12}^-)$, with $d_{12}^+\geq
|d_{12}^-|$. $O_3$ is used to map $L$ into the form $OD_{13}$, for diagonal
$D_{13}$ and some matrix $O\in SO(2)$. Thus, the standard form is given by
\bea\label{eq:stdform} %
\gamma_s =\left(\begin{array}{ccc}
    \lambda_1 \one &D_{12}& OD_{13} \\
    D_{12}& \lambda_2 \one& M\\ D_{13}O^T & M^T&\lambda_3
    \one\end{array}\right), 
\eea%
where $D_{12}$ and $D_{13}$ are diagonal and $O\in SO(2)$. Hence, the number
of free parameters in \Eqref{eq:stdform} is $12$. In case of degeneracy, more
of the off--diagonal blocks can be made diagonal, as explained above. Due to
Theorem~\ref{th:GLUequiv} we know that two states are GLU--equivalent iff
their standard forms (Eq.~(\ref{eq:stdform})) coincide.

\subsubsection{Standard form: Pure states}
If we specialize to \emph{pure} states, the CM must fulfill additional
constraints and the number of free parameters is greatly reduced. We then have
$\gamma J\gamma = J$, i.e., $\gamma$ is a symplectic matrix.  Taking into
account that $\gamma$ is symmetric, we have $\gamma = SS^T$ for a symplectic
matrix $S=ODO'$. The number of real parameters describing a pure $n$-mode
state is therefore $n^2+n$. Since the GLU, i.e., the local (single-mode)
symplectic operations are parameterized by $3n$ parameters, one would expect a
$n^2-2n$-parameter standard form. Hence, for the three mode Gaussian states
considered here, one would expect three free parameters. In order to derive
the parameterization we first show in the following theorem, that pure
three-mode Gaussian states are of a particularly simple form.

\begin{theorem}[$1\times1\times1$ pure state $xp$ block
  diagonal]\label{th:StdForm} 
  Any pure $1\times1\times1$ Gaussian state is GLU--equivalent to a state,
  whose CM, $\gamma$, as given in Eq.~(\ref{eq:1}) has the property that
  \emph{all} the submatrices $A,B,C,K,L,M$ are diagonal. I.e., in the
  $xp$-ordered basis we have
  \begin{align}\label{eq:2}
      \gamma_s &= \gamma_x\oplus\gamma_x^{-1}, \mbox{ where}\\
      \gamma_x &= \left( \begin{array}{ccc}
          \lambda_1 &d_{12}^+&d_{13}^+\\ d_{12}^+&\lambda_2 &d_{23}^+\\
          d_{13}^+&d_{23}^+ &\lambda_3 
\end{array} \right),
 \end{align}
 with $\lambda_i$ denoting the local purities and $d_{ij}^+ \in \R$.
\end{theorem}
\begin{proof}
  In Appendix~\ref{app:A} we show that the necessary condition for $\gamma$ to
  correspond to a pure state, $\gamma J\gamma= J$, implies that all
  submatrices, $K,L,M$ have to be diagonal. This implies that pure three-mode
  states can always be brought into a form in which correlations exist only
  among the $X$-quadratures and among the $P$-quadratures, respectively. That
  is, the CM is \emph{$xp$-blockdiagonal} in the standard form i.e.,
  $\gamma=\gamma_x\oplus\gamma_p$ (in the $xp$-ordered basis). Using then that
  the state is pure, which implies the condition $\tilde{J} \gamma \tilde{J}^T
  \gamma=\one$, where $\tilde{J}=[0_{n}, -\one_n;\one_n,0_m]$ is $J$ in the
  $xp$-basis, and $0_n,(\one_n)$ denote the $n\times n$ zero (identity)
  matrix, respectively, it is easy to see that for pure states
  $\gamma_p=\gamma_x^{-1}$, which proves the statement.
\end{proof}

Since the positive real and symmetric matrix $\gamma_x$, can always be written
as $\gamma_x=ODO^T$ for $O$ orthogonal and $D$ real and diagonal, six free
parameters are required to characterize $\gamma_x$. Since in the standard form
both $\gamma_x$ and $\gamma_p$ must have the same diagonal elements, this
yields three constraining equations, leaving 3 parameters characterizing the
equivalence classes. We derive in the next section the conditions on those
parameters to correspond to a valid CM of a pure state.

\subsubsection{Parameterization of pure $1\times 1\times 1$ states}
As we have just seen, an arbitrary pure 3--mode state can be written as
\begin{equation} \label{GammaDiag}
\gamma =\left(\begin{array}{ccc}
\lambda_1 \one &D_{12}& D_{13} \\
D_{12}& \lambda_2 \one& D_{23}\\ D_{13}  & D_{23}&\lambda_{3}
\one\end{array}\right),
\end{equation}
where $D_{ij}$ is diagonal. Due to the condition $\gamma\geq i J$ [see
Eq.~(\ref{iii})], we have $\lambda_i\geq 1$ $\forall i$.

In this section we derive the conditions for $\gamma$ corresponding to a pure
state and show that the CM can be fully parameterized by the three
local-mixedness parameters $\lambda_j$. Recall that $\gamma$ is pure iff
$\gamma \geq 0$ and $\gamma J \gamma=J$. We first derive the necessary and
sufficient conditions for a matrix $\gamma$, as given in Eq.~(\ref{GammaDiag})
with $\lambda_i\geq 1$ to fulfill $\gamma J \gamma =J$ (see Lemma
\ref{lemmapure1}). After that, we derive the condition for such a matrix to be
positive (see Lemma \ref{lemmapure2}).

\begin{lemma} \label{lemmapure1} A matrix $\gamma$, as given in
  Eq.~(\ref{GammaDiag}) with $\lambda_i\geq 1$ fulfills $\gamma J \gamma=J$
  iff the entries of the diagonal matrices
  $D_{ij}=\mbox{diag}(d_{ij}^+,d_{ij}^-)$ are given (up to GLUs) by 
\begin{align}
\label{Eq_entriesD}
d_{ij}^\pm&=\frac{1}{4\sqrt{\lambda_i\lambda_j}}(\sqrt{a_{ij}}\pm
\sqrt{b_{ij}}),
\end{align}
with
\begin{align}
\nonumber
a_{ij}&=[(\lambda_i-\lambda_j)^2-(\lambda_k-1)^2][(\lambda_i-\lambda_j)^2-(\lambda_k+1)^2] \\
\nonumber
b_{ij}&=[(\lambda_i+\lambda_j)^2-(\lambda_k-1)^2][(\lambda_i+\lambda_j)^2-(\lambda_k+1)^2],
\end{align}
where $i\not=j$ and $k\not= i,j$ refers to the third index. 
\end{lemma}
\begin{proof}
It is straight forward to show that the condition $\gamma J \gamma =J$ is equivalent to the following set of equations,
\begin{subequations}
\bea \label{Eqgammapure}
\lambda_1^2 + |D_{12}|+|D_{13}|=1 \label{mEq1}\\
\lambda_2^2 + |D_{12}|+|D_{23}|=1 \label{mEq2}\\
\lambda_3^2 + |D_{13}|+|D_{23}|=1 \label{mEq3}\\
\lambda_1 D_{12}+\lambda_2 \tilde{D}_{12}+ \tilde{D}_{13}\odot D_{23}=0\label{mEq4}\\
\lambda_1 D_{13}+\lambda_3 \tilde{D}_{13}+ \tilde{D}_{12}\odot D_{23}=0\label{mEq5}\\
\lambda_2 D_{23}+\lambda_3 \tilde{D}_{23}+ \tilde{D}_{12}\odot D_{13}=0\label{mEq6},
\eea
\end{subequations}
where $\odot$ denotes the componentwise multiplication (Hadamard
product). Here, we used the notation $D_{ij}=\mbox{diag}(d_{ij}^+,d_{ij}^-)$,
$\tilde{D}=\sigma_x D \sigma_x$ and that $DJ=J\tilde{D}$, (i.e.,
$\tilde{D}=-JDJ$) for any diagonal matrix $D$ and therefore $DJD=|D|J$. Note
that if $D=\mr{diag}(a,b)$, then $\tilde{D}=\mr{diag}(b,a)$.  In
Appendix~\ref{AppB} we show that those conditions (together with
$\lambda_j\geq1$) are satisfied iff the entries of the diagonal matrices
$D_{ij}=\mbox{diag}(d_{ij}^+,d_{ij}^-)$ are given (up to GLUs) by $d_{ij}^\pm$
as given in the lemma.
\end{proof}

Note that in \cite{ASI06} it has been stated that a generic state can be
written as in Eq.~(\ref{GammaDiag}), with the entries of the diagonal matrices
given in Eq.~(\ref{Eq_entriesD}). However, we are aiming here for a complete
characterization of three--mode pure states. As we prove below, the results of
\cite{ASI06} hold for all pure three-mode Gaussian states.

Clearly $a_{ij},b_{ij}$ must be positive in order to obtain a real CM. This
leads to the (mutually exclusive) conditions
$|\lambda_i-\lambda_j|\leq\lambda_k-1\,\forall (ijk)$ or
$|\lambda_i-\lambda_j|\geq\lambda_k+1\,\forall (ijk)$. We show now that only
the first condition is compatible with the positivity of the reduced CM (at
modes $(ij)$). To see that, note that for pure three-mode states it follows
from Eqs.~(\ref{mEq1}-\ref{mEq3}) that for all $(ijk)$:
$\lambda_k^2=\lambda_i^2+\lambda_j^2+2|D_{ij}|-1=(\lambda_i+\lambda_j+1)^2-2(\lambda_i+\lambda_j+\lambda_i\lambda_j-|D_{ij}|+1)$. The
last term in this expression is strictly negative since due to the fact that
the CM of the modes $i,j$ has to be positive, we have
$\lambda_i\lambda_j\geq\pm|D_{ij}|$, which implies that
$\lambda_k<\lambda_i+\lambda_j+1$. Thus, the conditions
\begin{equation}
  \label{CondPos} 
  \lambda_i+1\leq\lambda_j+\lambda_k\,\forall(ijk).
\end{equation}
are the necessary and sufficient conditions for a valid pure CM $\gamma$ to be
real. Note that if $\lambda_i \geq \lambda_j,\lambda_k$, the conditions in
Eq.~(\ref{CondPos}) are equivalent to the condition $\lambda_i \leq
-1+\lambda_j+\lambda_k$. For later reference, we also note the simple 
expression for $|D_{ij}|$ in terms of the $\lambda$'s:
\begin{equation}
  \label{eq:detDij}
  |D_{ij}|=\frac{1}{2}\left( \lambda_k^3+1-\lambda_i^2-\lambda_j^2 \right).
\end{equation}

It remains to impose the condition that $\gamma\geq0$. For this, we use the
following Lemma (Schur's complement), which is proven for
instance in \cite{GKLC01b}.
\begin{lemma} (Positivity of self-adjoint matrices)\\ \label{positivity}
A self-adjoint matrix
\begin{equation}
M = \left( \begin{array}{cc} A&C\\ C^\dagger &B
\end{array} \right),
\end{equation}
with $B>0$ is positive if
and only if
\begin{equation}\label{poscondb}
A-C\frac{1}{B}C^\dagger\geq0.
\end{equation}
\end{lemma}
Using this lemma we show that any CM $\gamma$ as given in
Eq.~(\ref{GammaDiag}) is positive in case the condition (\ref{CondPos}) is
satisfied, as stated in the following lemma.

\begin{lemma} \label{lemmapure2}
The symmetric matrix $\gamma$, as given in Eq.~(\ref{GammaDiag}) with $\lambda_k\geq 1$, for $k\in \{1,2,3\}$ is positive semidefinite if Eq.~(\ref{CondPos}) holds.
\end{lemma}

\begin{proof} 
Since $\gamma=\gamma_x\oplus \gamma_x^{-1}$ (see Theorem~\ref{th:StdForm}), we
have $\gamma >0$ iff $\gamma_x>0$. Using now Lemma \ref{positivity} and the fact that $\lambda_3>0$, we know
that the $3\times 3$ matrix $\gamma_x$ is positive iff the $2 \times 2$ matrix
\begin{equation} 
Y=\begin{pmatrix} \lambda_1 & d_{12}^+\\d_{12}^+ &
  \lambda_2\end{pmatrix}-\frac{1}{\lambda_3} \begin{pmatrix} d_{13}^+\\
  d_{23}^+\end{pmatrix} \begin{pmatrix} d_{13}^+& d_{23}^+\end{pmatrix}>0.
\end{equation}
Note that $Y>0$ iff $|Y|>0$ and $\tr(Y)>0$. Using that $\lambda_k\geq 1$ for all
$k$, tedious, but elementary calculations (see \Aref{AppC}) show that both
expressions are positive if the condition (\ref{CondPos}) holds.
\end{proof}
Combining Lemma \ref{lemmapure1} and Lemma \ref{lemmapure2} we obtain the
following theorem.

\begin{theorem}\label{th:3modeStd}
Any CM of a pure 3--mode Gaussian state can be written (up to GLUs) as
\begin{equation} \label{GammaDiagTh}
\gamma =\left(\begin{array}{ccc}
\lambda_1 \one &D_{12}& D_{13} \\
D_{12}& \lambda_2 \one& D_{23}\\ D_{13}  & D_{23}&\lambda_{3}
\one\end{array}\right),
\end{equation}
where $D_{ij}=\mbox{ diag } (d_{ij}^+, d_{ij}^-)$, with $d^\pm_{ij}$ given in
Eq.~(\ref{Eq_entriesD}).
\end{theorem}

Thus, the non--local properties of any pure three-mode Gaussian state are
completely characterized by the local mixedness parameters, $\lambda_i$, i.e.,
by the bipartite entanglement shared between each mode with the other two.
Recalling our discussion of LU-- and GLU--equivalence at the beginning of this
Section, we see that (like the pure bipartite Gaussian states) also the pure
$1\times1\times1$ Gaussian states are completely characterized by their
Schmidt coefficients across the (three) different bipartitions (which are in
one-to-one relation with the $\lambda_l$). Therefore, those states are
LU--equivalent iff they are GLU--equivalent and Theorem~\ref{th:3modeStd} also
characterizes the LU--classes of pure three-mode Gaussian states.

\subsubsection{Some special cases}
\label{Sec_SpecialCase}
Let us briefly consider two special cases, namely the one where one of the
off-diagonal matrices, say: $\gamma_{ij}$, is (a) not invertible or (b)
proportional to $\id$.

Case (a): The condition $|D_{ij}|=0$ together with \Eqref{eq:detDij} implies
$\lambda_k^2=\lambda_i^2+\lambda_j^2-1$ and inserting it in
\Eqref{Eq_entriesD} we find that $d_{ij}^-=0$ and
$d_{ij}^+=\sqrt{(\lambda_i^2-1)(\lambda_j^2-1)(\lambda_i\lambda_j)^{-1/2}}$. Note
that $d_{ij}^+\not=0$ (as are $|D_{ik}|,|D_{jk}|$) unless $\lambda_i$ or
$\lambda_j$ equals $1$, in which case the respective mode factorizes and the
remaining two would be in a two-mode squeezed state.

Case (b): $D_{ij}\propto\id$ is only possible if $b_{ij}=0$, which implies that $\lambda_k=\lambda_i+\lambda_j-1$ (i.e., in particular, $\lambda_k\geq
\lambda_i,\lambda_j$ \footnote{Note that
$D_{ij}\propto\id$ implies that the reduced state of modes $i$ and $j$ has a positive semidefinite partial transpose and is therefore separable.}). Clearly, then
$\lambda_k-\lambda_{i(j)}=\lambda_{j(i)}-1$ and thus $a_{ik}=a_{jk}=0$. Hence,
the remaining two off-diagonal blocks are both proportional to $\sigma_z$. It also holds that if $D_{ij}\propto\sigma_z$, which
implies that $a_{ij}=0$ (which fixes $\lambda_k=1+|\lambda_i-\lambda_j|$), that one of the remaining two off-diagonal blocks is degenerate (and the other
$\propto\sigma_z$): If $\lambda_i\geq\lambda_j$ then $b_{jk}=0$ and $a_{ki}=0$
otherwise reversed.  As we see below, these states can all be generated by
letting a beam splitter couple one mode of a two-mode system in a two-mode
squeezed vacuum state with a third mode in the vacuum.

Another interesting special case are the fully permutation \emph{symmetric}
states \cite{LF02,ASI06}, for which the three local mixednesses are identical,
i.e., $\lambda_l=\lambda$ $\forall l$. We denote the CM of a symmetric state
in standard form by $\gamma_\mr{sym}(\lambda)$. These states 
were sometimes
called maximally entangled \cite{AdIl06,ASI06} due to their extremal
entanglement properties reminiscent of their qubit analogs \cite{DVC00}. For
these states the matrices $D_{ij}$ are given by $\diag(d^+,d^-)$ with
\begin{align}
  \label{eq:symmstates}
  d^\pm&=\frac{1}{4\lambda}\left(
    (\lambda^2-1)\pm \sqrt{9\lambda^4 -10\lambda^2+1} \right).
\end{align}
In Sec. \ref{Sec_SymInitial} we investigate which states can be obtained from
symmetric states via Gaussian Local Operations assisted by Classical
Communication (GLOCC).

\section{Generation of three-mode pure states}
\label{SecGeneration}
Let us briefly remark on the generation of pure three-mode states.  In
\cite{Adesso2007} a general state-generation scheme for this case is
presented. There, a two-mode squeezed state (of modes 1 and 2, with squeezing
parameter $r$) is coupled to mode 3 (in the vacuum state) by a sequence of
three beam splitters (BS) acting on modes (13), (23), and (13),
respectively. The transmissivities of the third BS is fixed while those of the
first two are adjusted such as to produce the desired local purities.

Note that in
the special case in which one of the off-diagonal
matrices is degenerate (case (b) above), a simplified scheme suffices: Letting
a beam splitter
with transmissivity $\cos^2\theta\in[0,1]$ act on
part of a two-mode squeezed vacuum
(with squeezing parameter $s\geq0$) and a vacuum mode allows to generate all
states with degenerate CM: If $\lambda_1$ is the largest local mixedness, then
\begin{equation}
  \label{eq:degState}
\gamma(s,\theta)=B(\theta)\left[ \gamma_\mr{tms}(s)\oplus\id
\right]B(\theta)^T,
\end{equation}
where
\[B(\theta)=\id\oplus \left( \begin{array}{cc}
    \cos\theta\id_2&\sin\theta\id_2\\ -\sin\theta\id_2&\cos\theta\id_2
  \end{array} \right)\]
and $\gamma(s,\theta)$ then has the three local purities
$\lambda_1=\cosh s, \lambda_2=\sin^2\theta+\cos^2\theta\cosh s,
\lambda_3=\cos^2\theta+\sin^2\theta\cosh s$,
satisfying the characteristic equation $\lambda_1+1=\lambda_2+\lambda_3$
of Case (b) above. And since for any given $\lambda_2,\lambda_3\geq1$ there is
a pair $(s,\theta)\in\RR^+\times[0,2\pi]$ such that the above equations
hold, we can generate all degenerate states this way. Since these states are
obtained from a two-mode squeezed state by distributing it via a beam splitter
among two parties, we also refer to them as \emph{distributed two-mode squeezed
  states}. 

In order to see how the different GLU classes relate to each other we now
extend the set of operations from Gaussian local unitaries to Gaussian
(stochastic) local operations with classical communication (GLOCC). In
particular, this will allow us to investigate whether the GHZ/W states are
\emph{maximally entangled} also in the sense that they allow to prepare any
other Gaussian state via GLOCC (in the same way as, e.g., the Bell state does
for two qubits or certain families of states do in the pure multi-qubit
setting \cite{VSK13}).

\section{Gaussian Local Operations}
\label{SecGLOCC}

LU equivalence leads to a very detailed classification of multipartite states
with a continuum of inequivalent classes. A more coarse-grained picture
emerges if interconvertibility of states under a larger class of
transformations, stochastic local operations and classical communication
(SLOCC) \cite{BPR+99} is studied. SLOCC plays an important role in
entanglement theory \cite{Nie99,Vid99b,DVC00}. Two states are said to by
SLOCC-equivalent if there is a non-zero probability to convert them into each
other. Due to the stochastic interconvertibility of all pure bipartite states
of equal Schmidt rank \cite{DVC00} there are $d-1$ different kinds of
bipartite (pure state) entanglement of $d$-dimensional systems. In contrast,
in the tripartite case, even for three qubits two inequivalent classes have
been identified that are not connected by SLOCC transformations \cite{DVC00}.

Also in the Gaussian setting, GLU operations can be extended by allowing for
local (generalized) measurements, namely adjoining additional modes (in a pure
state) and then performing (partial) Gaussian measurements.  However, Gaussian
SLOCC have not been investigated since the only Gaussian operators with a
bounded inverse are the Gaussian unitaries \footnote{However, one might note
  that the Gaussian operator $G_{-c}=e^{- c(X^2+P^2)}, c>0$ has an ``inverse''
  $G_{+c}=e^{+ c(X^2+P^2)}$ which is unbounded, but well-defined on the image
  of $G_{-c}$ (i.e., on sufficiently fast decaying states). This might be used
  to argue for the ``Gaussian SLOCC equivalence'' of, e.g., two-mode squeezed
  states with different squeezing parameters $r<r'$ (which can be
  ``stochastically'' mapped to each other by are related by $G_{\pm c}$ with
  $e^c=\frac{\tanh r'}{\tanh r}$). However, such unbounded filtering
  operations have no clear physical implementation and we do not pursue this
  further here.}.  Instead, we are interested here in the convertibility of
states under Gaussian LOCC (GLOCC). In light of the complicated structure of
general LOCC transformations \cite{CLM+12} the Gaussian case is remarkably
simple: all Gaussian operations can be characterized via the
Choi-Jamio{\l}kowski (CJ) isomorphism by an equivalent Gaussian state
\cite{GiCi02,Fiu02,Holevo2011}. When acting on a Gaussian state with
\emph{known} CM, all such transformations can be implemented
\emph{deterministically} by teleporting through that state
\cite{GiCi02}. While teleportation is probabilistic (yielding a random
displacement), this can be computed from the measurement outcome and the
involved CMs and can then be undone by local unitaries. In particular, this
implies, that a finite number of communication rounds is enough to implement
any GLOCC. Note that the inverse of a GLOCC is not Gaussian, hence GLOCC does
not induce an equivalence relation among Gaussian states but rather gives rise
to a partial ordering (see \cite{GECP03} for the bipartite case).

Gaussian operations mapping pure states to pure states (``pure operations'')
are characterized by a pure CJ-CM $\Gamma$ and pure operations on a single
mode are characterized by a pure $1\times1$ CM, i.e., by one GLU-invariant
parameter $r$ (two-mode squeezing) and two sets of three local parameters
(each characterizing a single-mode Gaussian unitary), which describe local
unitary pre- and post-processing of the state. Following the treatment in
\cite{GECP03} for the bipartite case we can easily obtain expressions for the
output CM of a three-mode state after a general three-mode GLOCC.

\subsection{General GLOCC on three-mode systems}

The most general Gaussian operation transforming a three mode Gaussian state
into another, corresponds to a six mode CJ-CM
$\Gamma=[\Gamma_1,\Gamma_{12};\Gamma_{12}^T, \Gamma_2]$. Here, the index 2 (1)
denotes the three input (output) modes respectively. According to
\cite{GiCi02} the output CM $\gamma'$ is related to the input CM $\gamma$ by
\begin{equation}
  \label{eq:GLOCCgeneral}
  \gamma'=\Gamma_1-\Gamma_{12}\left( \Gamma_2 + \Lambda\gamma\Lambda
  \right)^{-1}\Gamma_{12}^T,
\end{equation}
where $\Lambda=\oplus_{x=1}^3\sigma_z$.
For ease of notation we
denote the three diagonal blocks of the input-CM $\gamma$ by $A_x,
x=1,2,3$ and use the convention that indices $(x,y,z)$ in a single equation
refer to distinct modes. In the case of pure LOCC transformations, the CM
$\Gamma$ is block diagonal, i.e.
$\Gamma = \oplus_{x=1}^3\Gamma_x$ with
\begin{align}
\Gamma_x&=\left( \begin{array}{cc}\Gamma_{1x}&\Gamma_{12x}\\
\Gamma_{12x}^T&\Gamma_{2x}
\end{array} \right)\\
&\equiv (S_{1x}\oplus S_{2x})\gamma(r_x)(S_{1x}\oplus S_{2x})^T.
\end{align}
Using that $S_{1x}$ only describes local unitary postprocessing (which is
irrelevant for GLU-invariant properties) we can without loss of generality
take $S_{1x}=\id$. We write the Euler decomposition \cite{ADMS95b} of
$S_{2x}=O_{1x}Q_xO_{2x}$ where $O_{1x},O_{2x}$ are in $\mr{SO}(2)$ and
$Q_x=\mr{diag}(q_x,q_x^{-1})$ with $q_x>0$ is a single-mode squeezing
  transformation. Since the effect of $O_{2x}$ can be undone by local unitary
  postprocessing, we set $O_{2x}=\id$ and obtain
\begin{align}
  \label{eq:3}
  \Gamma_{1x}&=\cosh r_x\id,\\
\Gamma_{2x}&=\cosh r_x O_{1x}Q_x^2O_{1x}^T,\\
\Gamma_{12x}&=\sinh r_x \sz Q_xO_{1x}^T,
\end{align}

To make the ensuing expressions shorter, we will from now on use the
  notation $\cosh r\equiv \ch r$ and $\sinh r\equiv \sh r$.\\
 Using the
Schur-complement formula for the inverse of a symmetric matrix, this allows to
write the reduced CMs of the output state $\gamma'$ at mode $x$ in compact
form as
\begin{align}
  \label{eq:4}
(\gamma')_{xx}&=
\Gamma_{1x}-\Gamma_{12x}\left( \Gamma_{1x}+A_x-T_{x}\right)^{-1}\Gamma_{12x}^T,\\
\end{align}
where we have introduced the auxiliary matrices
\begin{align}
  \label{eq:5}
T_x &=\left( D_{xy}\,\,\, D_{xz} \right) R_x^{-1} \left(
  \begin{array}{c} D_{xy}\\ D_{xz}
\end{array}\right),\\
  \label{eq:6}
  R_x &=
\left( \begin{array}{cc}
\Gamma_{2y} + A_y & D_{yz}\\
D_{yz} & \Gamma_{2z} + A_z
\end{array} \right).
\end{align}
Note that the identity operation corresponds to the limiting case
of infinitely squeezed CJ-CM $\Gamma$ (i.e., $r\rightarrow \infty$ and
$O_1=X=\one$). Hence, the case of not operating on mode $x$ corresponds to
taking the limit $r_x\rightarrow \infty$ in the above expressions. Since the
expressions in terms of the 9 pure GLOCC and 3 CM parameters is rather long
and intransparent, we split the general three-mode GLOCC into a sequence of
three single-mode transformations. Sometimes we will focus on a much simpler
family of transformations [which we refer to as (local) TMS filtering], namely
those, where $O_{1x}=Q_x=\id\forall x$ (i.e., no unitary pre-processing),
leaving only the three two-mode squeezing parameters free. In the bipartite
case, GLOCC of that form are know to suffice to perform all possible
transformations between GLU equivalence classes.

Here, our aim is not the complete analysis of the GLOCC-transformations of
three-mode pure Gaussian states but only an illustration of the usefulness of
the GLU classification and standard form. In particular, we use the standard
form derived in the previous section to study which pure three-mode states,
can be transformed into each other by GLOCC. We first provide simple
expressions for the single-mode transformations of three-mode states which we
then use in the subsequent sections to show that the CV GHZ/W states lack
certain properties of maximal entanglement. To this end we first show that
there are pure Gaussian three-mode states that cannot be obtained from any
$\gamma_\mr{sym}$ via GLOCC by identifying a qualitative feature that the
symmetric states lack and that cannot be generated by GLOCC. Then we identify a one-parameter family of such states
(unreachable from $\gamma_\mr{sym}(\lambda)$) that, in contrast, allows to
reach all symmetric states.

\subsection{GLOCC-transformation of $1\times1\times1$ states}

As we have seen before, a three-mode pure Gaussian state is completely
characterized by its local mixedness parameters, $\lambda_i$. Therefore, we
simply write $(\lambda_1,\lambda_2,\lambda_3)$ when referring to the CM given
in Eqs.~(\ref{GammaDiag},\ref{Eq_entriesD}). Here we derive a compact
prescription of how the CM of a Gaussian state changes under single-mode GLOCC
and, in particular, give expressions for the matrices determining the local
mixedness parameters $\lambda_i$.

Let us denote the $1\times2$ CM of
the three-mode Gaussian input state $(\lambda_1,\lambda_2,\lambda_3)$ by
\begin{equation}\label{eq:inputCM}
\gamma = \left( \begin{array}{cc} A & K\\
K^T & R
\end{array} \right),
\end{equation}
where $A=\lambda_i \one$, $R=\left( \begin{array}{cc} A_j&D_{jk}\\
    D_{jk}&A_k
\end{array} \right)$, with $A_l=\lambda_l \one$, and
$K=(D_{ij}\,\, D_{ik})$, where the block $A$ refers to the mode $i$ to be
acted upon.

As mentioned above, Gaussian completely positive maps (CPMs) acting on a
single mode and mapping pure states to pure states are in one-to-one
correspondence to pure two-mode Gaussian states by the CJ isomorphism
\cite{Cho72,Jam72,GiCi02,Holevo2011}; they are GLU-equivalent to a
two-mode squeezed state and can therefore be completely characterized by the
$1\times1$ CM $\Gamma$  [see Eq.~(\ref{auxgamma2})]
\begin{equation}
  \label{eq:1modeCPM}
  \Gamma = \left(\begin{array}{cc} \Gamma_1&\Gamma_{12}\\ \Gamma_{12}^T&\Gamma_2 \end{array}\right)\equiv
(S_1\oplus S_2)\gamma(r)(S_1\oplus S_2)^T,
\end{equation}
with symplectic $S_1,S_2$. As discussed in the previous section, we can
without loss of generality choose $S_1=\id$ and $S_2=O_1^TX^{-1}$ with $X=\diag(x,x^{-1})$.

If the CPM corresponding to $\Gamma$ acts on mode $i$ of the state with CM
$\gamma$ of \Eqref{eq:inputCM}, it is transformed to
$\gamma'$ with \cite{GECP03}
\begin{align}
\label{eq:Aprime}
A'&=\Gamma_1-\Gamma_{12}(\Gamma_2+\sz A\sz)^{-1}\Gamma_{12}^T,\\
\label{eq:Rprime}
R'&=R-K^T\sz(\Gamma_2+\sz A\sz)^{-1}\sz K,\\
\label{eq:Kprime}
K'&=\Gamma_{12}(\Gamma_2+\sz A\sz)^{-1}\sz K.
\end{align}
To characterize the output state only the three $2\times 2$
diagonal blocks of $\gamma'$ are of interest. We have:
\begin{align}
\label{GLOCCtrA}
  A_i'&=\ch r\id -\sh^2r(\ch r\id +\lambda_iX^2)^{-1},\\
\label{GLOCCtrB}
A_j'&= \lambda_j\id - T_j(\ch r\id+\lambda_iX^2)^{-1}T_j^T,\\
\label{GLOCCtrC}
A_k'&= \lambda_k\id - T_k(\ch r\id+\lambda_iX^2)^{-1}T_k^T,
\end{align}
where $T_l=D_{il}O_1X$ for $l=j,k$. Clearly, up to GLU the final state depends
only on the parameters $r,x,\phi$, where $O_1=e^{i\phi\sigma_y}$. Note that
these expressions could be obtained from Eq.~(\ref{eq:4}) in the
limit $r_y,r_z\to\infty$.

We now use the GLOCC formalism to explore the entanglement properties of
certain families of pure three-mode Gaussian states, in particular the
symmetric states $\gamma_\mr{sym}(\lambda)$. For large $\lambda$ these are
highly entangled states and they have been referred to as 
as maximally entangled continuous-variable states. We show,
however, in the next subsection that in contrast to what one might expect, it
is \emph{not} possible to prepare by GLOCC an arbitrary pure three-mode
Gaussian state from a state $\gamma_\mr{sym}(\lambda)$ no matter how large
$\lambda$. In contrast, we study in the final subsection a different family,
and show that it allows to prepare, in particular, all symmetric states.

\subsection{Symmetric initial states}
\label{Sec_SymInitial}
We show now that it is not possible to reach an arbitrary three-mode entangled
state via GLOCC from a symmetric three-mode entangled state. To this end, we first show that a state can be generated from a symmetric
state $\gamma_\mr{sym}$ by a single-mode
GLOCC if and only if it can be generated (possibly from some other symmetric
initial state)  via a single-mode TMS operation [i.e., $X=O_1=\one$ in
Eqs.~(\ref{GLOCCtrA}-\ref{GLOCCtrC})] and that any state
$(\lambda_1',\lambda_2',\lambda_2')$ with
$\lambda \geq \lambda_2'\geq \lambda_1'$ can be generated in this way.
Then we show that starting with a state $(\lambda_1',\lambda_2',\lambda_2')$ a
general measurement on the second mode allows only to reach states with
$|D_{12}|=(\lambda_3'^2-\lambda_1'^2-\lambda_2'^2 +1)/2 \leq 0$. After that,
we show that performing a general GLOCC on the third mode cannot change the
sign of this determinant. Consequently, a pure three-mode Gaussian state with
$|D_{12}|>0$ cannot be prepared by general GLOCC starting from an (arbitrary)
symmetric Gaussian state. In order to show that it is not in general the case
that the sign of the determinants of the off--diagonal blocks cannot be
changed via GLOCC, we present in the subsequent subsection a class of states
with one positive and two negative determinants, from which states with three
negative determinants can be obtained via GLOCC.

Let us first show that from a symmetric state with parameter $\lambda$ and
operating on mode 1 only, we obtain $(\lambda_1',\lambda_2',\lambda_2')$ with
$\lambda\geq\lambda_2'\geq\lambda_1'$. Then we show that \emph{any} ratio
  $\lambda_2'/\lambda_1'\geq1$ can be obtained by a suitable TMS-operation
  and suitable choice of the initial $\lambda$.\\
  From \Eqref{GLOCCtrA} we see that $\lambda_1'$ does not depend on $O_1$ and
  takes a global maximum for $x=1$. Since a TMS-operation yields
  \begin{align}\label{eq:l1p}
  \lambda_1'&=\frac{\lambda\ch r+1}{\lambda+\ch r},
\end{align}
which can take all values in $[1,\lambda]$ restricting to these operations
does not constrain $\lambda_1'$.\\
With \Eqref{GLOCCtrB}, one readily checks that $\lambda_2'$ is minimal
\footnote{$\lambda_2'$ is periodic with the angle $\phi$ with period $\pi$ and
  monotonically increases in the interval $[0,\pi]$.}  for $O_1=\id$. Thus
$\phi\not=0$ only increases the ratio $\lambda_2'/\lambda_1'$. Since as we see
below \emph{all} such ratios $\geq1$ can be obtained by TMS operations, we can
set $\phi=0$.  Looking now at $\lambda_2'/\lambda_1'$ for the case $\phi=0$,
we easily see that it is $\geq1$ \footnote{We have (with $c_x=x^2+x^{-2},
  s_x=x^2-x^{-2}$): 
  \[\frac{\lambda_2'^2}{\lambda_1'^2} = \frac{\lambda^2(1+\ch^2
    r)+\frac{3\lambda^4+6\lambda^2-1}{8\lambda}\ch
    rc_x+\frac{\lambda^2-1}{8\lambda}\sqrt{9\lambda^4-10\lambda^2+1}s_x}{\lambda^2\ch^2 
    r+c_x\lambda\ch r}.\] The
  difference of numerator and denominator is
  $\lambda^2-1+\frac{(\lambda^2-1)(3\lambda^2+1)}{8\lambda}\ch r
  c_x+\frac{(\lambda^2-1)\sqrt{9\lambda^4-10\lambda^2+1}}{8\lambda}\ch r s_x$,
  which is positive since $c_x\geq|s_x|$ and
  $(3\lambda^2+1)^2>9\lambda^4-10\lambda^2+1$.}  Note that this is expected
since the GLOCC (a partial measurement) is performed at mode 1 and thus our
lack of knowledge about the local state there
is less than at the unmeasured modes.\\
To complete the proof we have to show that all such ratios can be achieved by
TMS operations. The parameter $\lambda_2'$ after such a GLOCC is
\begin{align}\label{eq:l2p}
\lambda_2'&=\frac{\left[
\lambda^2(\ch^2 r+1)+\frac{3\lambda^4+6\lambda^2-1}{4\lambda}\ch
r\right]^{1/2}}{\lambda+\ch r}.
\end{align}
That $\lambda_2'\geq\lambda_1'$ is easily seen using that $\lambda\geq 1$,
which implies that the first term in the numerator is larger than or equal to
$\lambda^2 \ch r^2 +1$ and the second term is larger than or equal to $2
\lambda \ch r$.

To see that \emph{all} such pairs $(\lambda_1',\lambda_2')$ can be
achieved by suitable choice of the initial parameter $\lambda$ and
operation-parameter $r$
we can invert Eqs.~(\ref{eq:l1p},\ref{eq:l2p}) to find $r,\lambda$ as a function of the target
values $\lambda_1'$ and $f\geq0$ which determines the ratio
$\lambda_2'/\lambda_1'$
via
\[
\left( \frac{\lambda_2'}{\lambda_1'} \right)^2=1+f.
\]
We find
\begin{subequations}
    \label{eq:10}
\begin{align}
\lambda &= \left[
  \frac{(3+4f)\lambda_1'^2+1}{6\lambda_1'}+\sqrt{\left(\frac{(3+4f)\lambda_1'^2+1}{6\lambda_1'}\right)^2+\frac{1}{3}}\right],\\
  \ch r &= \frac{\lambda\lambda_1'-1}{\lambda-\lambda_1'}.
\end{align}
\end{subequations}
One readily checks that the values of $\lambda, r$ are in the admissible range
($\lambda\geq1,r\geq0$) for all valid target values $\lambda_1'\geq1,
f\geq0$, which proves the statement.

However, it is not possible to obtain all pure three-mode Gaussian states from
a symmetric initial state, not even by the most general GLOCC. This follows
from the fact that the symmetric states all have the property that the three
off-diagonal matrices $D_{ij}$ all have non-positive determinants
$|D_{ij}|\leq0$. As we show in the following Lemma, this is a property that
cannot be changed by GLOCC. However, we have already encountered states [such
as the distributed two-mode squeezed states $\gamma(s,\theta)$,
cf.\ \Eqref{eq:degState}] which have one non-positive determinant. These,
therefore, cannot be reached by GLOCC from the symmetric states.

\begin{lemma}
  It is impossible with GLOCC to transform a pure three-mode Gaussian state
  with three non-positive determinants $|D_{ij}|\leq0$ into a state with at
  least one (strictly) positive determinant.
\end{lemma}

\begin{proof}
  We consider an arbitrary initial state with $|D_{ij}|\leq0$ for all $(ij)$,
  i.e., cf.~\Eqref{eq:detDij}, $\lambda_i^2-\lambda_j^2-\lambda_k^2 +1
  \leq0\,\forall (ijk)$ and apply an arbitrary measurement on the $k$th
  mode. Without loss of generality, we choose $(ijk)=(123)$. Similarly to
  before, we obtain for the matrices in the diagonal of $\gamma$
  \begin{align}
A_1'&=\lambda_1\one -   T_1(\ch{r}\one + \lambda_3 X^2)^{-1}T_1^T\\
A_2'&= \lambda_2 \one -
T_2(\ch{r}\one + \lambda_3 X^2)^{-1}T_2^T,\\
A_3'&= \ch{r} \one - \sh{r}^2 (\ch{r}\one + \lambda_3 X^2)^{-1}
\end{align}
with $T_1 = D_{13} O_1 X$ and $T_2 = D_{23} O_1X $. Again, we consider the term
  $C_3\equiv\lambda_3'^2-\lambda_2'^2-\lambda_1'^2+1$, which now yields a more lengthy expression
  \begin{equation}
    \label{eq:C3}
C_3=\frac{[c_xA+Bs_x\cos(2\phi)]\ch r+C}{4\lambda_3(\lambda_3^2+\ch^2
  r+c_x\lambda_3\ch r)}
  \end{equation}
  \begin{align}
    \label{eq:7}
    A&=\lambda_1^4-2 \left(\lambda_2^2+\lambda_3^2+1\right) \lambda_1^2+\lambda_3^4+\left(\lambda
   _2^2-1\right)^2\\ \nonumber
&{}-2 \left(\lambda_2^2-3\right) \lambda_3^2,\\
B&=\left[\left(\lambda_1-\lambda_2-\lambda_3-1\right)
  \left(\lambda_1-\lambda_2-\lambda_3+1\right)\times\right.\\ \nonumber
&{}\times \left(\lambda_1+\lambda_2-\lambda_3-1\right)
\left(\lambda_1+\lambda_2-\lambda_3+1\right)\times\\ \nonumber
&{}\times \left(\lambda_1-\lambda_2+\lambda_3-1\right)
\left(\lambda_1-\lambda_2+\lambda_3+1\right)\times\\ \nonumber
&{}\left.\times \left(\lambda_1+\lambda_2+\lambda_3-1\right)
   \left(\lambda_1+\lambda_2+\lambda_3+1\right)\right]^{\frac{1}{2}},\\
C&=4 \lambda_3(\ch r^2+1) (\lambda_3^2-\lambda_1^2-\lambda_2^2+1),\\
\label{eq:cx} c_x&=x^2+x^{-2},\\
s_x&=x^2-x^{-2}.
  \end{align}
  Since the denominator, $\ch r$, and $B$ are positive \footnote{Note that $B$
    is equivalent to the square root of $w_1$ given in Appendix C, which is
    necessarily positive.}, to maximize this expression, $\cos2\phi$ should
  have maximal modulus and the same sign as $s_x$; i.e., without loss of generality we can take
  $x>1$ and $\phi=0$. Note also, that $C<0$ by assumption since the state $(\lambda_1,\lambda_2,\lambda_3)$ has $|D_{12}|<0$. To show finally that the whole numerator is
  always negative, we consider the expression for the other two determinants
  $|D_{ij}|$, namely $C_2=\lambda_2'^2-\lambda_1'^2-\lambda_3'^2+1$ and
  $C_1=\lambda_1'^2-\lambda_2'^2-\lambda_3'^2+1$, which are
\begin{align}
  \label{eq:C1}
  C_1&=\frac{\left(\lambda_1^2-\lambda_2^2-\lambda_3^2+1\right) x^2 \sh^2r}{\left(\lambda_3+x^2 \ch r\right) \left(\ch r+\lambda_3 x^2\right)},\\
  \label{eq:C2}
C_2&=\frac{\left(\lambda_2^2-\lambda_1^2-\lambda_3^2+1\right) x^2 \sh^2r}{\left(\lambda_3+x^2 \ch r\right) \left(\ch r+\lambda_3 x^2\right)},
\end{align}
i.e., up to a positive factor they are given by the \emph{input} determinants
$|D_{23}|$ and $|D_{13}|$, which thus do not change sign and remain non-positive.
Clearly $C_1+C_2+C_3=3-\lambda_1'^2-\lambda_2'^2-\lambda_3'^2\leq0$. This relation
must hold for all choices of $x$ and $r$. Now consider the limit $x\to\infty$,
for which both $C_1,C_2\to0$, and
\begin{equation}
  \label{eq:C3limit}
  C_3\to\frac{A+B}{4\lambda_3^2},
\end{equation}

which must therefore be $\leq0$. Hence $A+B\leq0$, therefore $(Ac_x+Bs_x)\ch
r+C\leq (A+B)c_x\ch r+C\leq C<0$ which shows that all three determinants
remain non-positive. This proves that
  a single-mode GLOCC does not allow to transform a state with only
  non-positive off-diagonal
  determinants $|D_{ij}|<0$ into a final state with at least one positive
  determinant. Since any pure GLOCC operation is represented by a product
  CJ-state and can therefore be decomposed in a sequence of three single-mode
  operations, we have shown even the most general pure GLOCC cannot achieve
  this.
\end{proof}

Thus, in particular, we have shown that from a
  symmetric state $\gamma_\mr{sym}(\lambda)$, which has
  $|D_{ij}|=-(\lambda^2-1)/2\leq0$, it is impossible to obtain via
arbitrary GLOCC any state with $|D_{ij}|>0$ for some $(ij)$.

\subsection{Initial states with positive determinant}
To show that these signs are not a GLOCC-invariant, and that, in fact, a
positive determinant $|D_{ij}|>0$ can always be made negative by GLOCC we
prove the following
\begin{lemma}
  Given a pure three-mode Gaussian state with one positive determinant
  $|D_{ij}|>0$, there exists a GLOCC to transform it into a state with three
  negative determinants.
\end{lemma}
\begin{proof}
Recall that for pure three-mode states there is at most one positive
determinant, see, e.g., Eq.~(\ref{Eqgammapure}). Assume, without loss of
generality, that $|D_{12}|>0$. From 
Eqs.~(\ref{eq:C1},\ref{eq:C2}) it is clear that to change
the sign of $|D_{12}|$ we must perform a GLOCC at mode 3. The determinant
after a general one-mode GLOCC is given by \Eqref{eq:C3}. Now consider the
case $\phi=0$ and the limits $x\to\infty$ and $x\to0$. As before, the limit
$x\to\infty$ proves that $A+B\leq0$, cf.~\Eqref{eq:C3limit}. For $x\to0$, we
obtain
\begin{equation}
  \label{eq:8}
  C_3(x\to0)= \frac{A-B}{4\lambda_3^2}.
\end{equation}
Since $B>0$ it follows that $C_3(x\to 0)<0$, i.e., for sufficiently small $x$
all three determinants $C_i$ are negative.
\end{proof}

Let us, as an example consider the distributed two-mode squeezed states with
CM $\gamma(s,\theta)$, discussed in \Secref{SecGeneration}. They are obtained
by passing part of a two-mode squeezed state $\gamma(s)$ through a beam
splitter with transmissivity $t=\cos^2\theta$, see \Eqref{eq:degState}. These
states have one off-diagonal block proportional to the identity, say
$D_{23}=-\sin\theta\cos\theta(\ch s-1)\id$, i.e., with positive
determinant. The other two off-diagonal blocks are proportional to $\sigma_z$,
i.e., $D_{12}=\cos\theta\sh s\sigma_z$ and $D_{13}=-\sin\theta\sh
s\sigma_z$. When performing a GLOCC characterized by two-mode squeezing
parameter $r$ and local squeezing $x$, i.e.,
$\Gamma=[\id\oplus\diag(x,x^{-1})]\gamma(r)[\id\oplus\diag(x,x^{-1})]$ on mode
1 (the one with the largest local mixedness), we obtain from
Eq.~(\ref{eq:Rprime}):
\begin{equation}
  \label{eq:D23p}
D_{23}'=D_{23}-D_{12}\left[ \ch r\left( \begin{array}{cc}x^2&0\\ 0&x^{-2}
\end{array} \right) +\ch s\id\right]^{-1}D_{13},
\end{equation}
therefore
\begin{align}
  \label{eq:D23entries}
  d_{23}^{\pm'}&= -\frac{\sin(2 \theta) \sh^2(s/2)(x^{\pm2}\ch r-1)}{\ch s +
    x^{\pm2}\ch r},
\end{align}
i.e., for $x^2<\ch r$ or $x^{-2}<\ch r$ one of the two coefficients is
negative (while the other is positive), yielding $|D_{23}'|<0$ for all
$x\not\in[\sqrt{1/\ch r},\sqrt{\ch r}]$. Since the signs of $|D_{12}|$ and
$|D_{13}|$ do not change, we have transformed the initial state with
$\mbox{sign}(|D_{12}|)=-1$, $\mbox{sign}(|D_{13}|)=-1$, and
$\mbox{sign}(|D_{23}|)=1$ to a state with all signs negative.

In fact, we can even obtain all 
symmetric states starting from
a distributed two-mode squeezed initial state. Let us consider the simple
one-parameter family of degenerate states with CM 
$\gamma=\gamma(s,\theta=\pi/4)$. Clearly, in this case the two smaller of the
three parameters are identical, i.e., our initial state is $\lambda_1=\ch s$
and $\lambda_2=\lambda_3=(\ch s+1)/2$. Then it suffices to perform a suitable
measurement at mode 1 to obtain $\lambda_1'=\lambda_2'=\lambda_3'$. Moreover,
by choosing $s$ large enough, it is possible to obtain \emph{all} 
symmetric states
this way. To see this, we use again Eqs.~(\ref{GLOCCtrA}-\ref{GLOCCtrC}) for
an operation characterized by $(r,x,\phi=0)$. Then it is straight forward to
see that by taking
\begin{equation}
  \label{eq:chx}
  x^2+x^{-2}=\frac{\csch^2\frac{s}{2} \left(3 \ch^2 r \ch s +\ch^2 r-\ch
   s-3\right)}{2 \ch r},
\end{equation}
we can prepare the symmetric state $\gamma_\mr{sym}(\lambda')$ with
\begin{equation}
  \label{eq:13}
  \lambda'= \frac{4 \ch^2\frac{s}{2} \left(\ch^2r \ch s-1\right)}{6
   \left(\ch^2 r-1\right) \ch s-2 \ch^2r+\ch(2s)+1}.
\end{equation}
Note that the right-hand side of \Eqref{eq:chx} is $\geq2$ for all choices of $s,r$, thus
there is always $x\geq1$ corresponding to the desired value. Considering the
limits $s\to\infty$ and $s\to0$, we see that $\lambda'\to\ch^2 r$ and
$\lambda'\to1$, respectively. Consequently, for any target state
$\gamma_\mr{sym}(\lambda')$ there exist $s,r\geq0$ and $x>0$ such that the
symmetric state $\lambda'$ can be prepared from the degenerate state
$\gamma(s,\pi/4)$ by a the single-mode GLOCC with parameters $(r,x,0)$. Thus,
the one-parameter family $\{\gamma(s,\pi/4):s\geq0\}$ is ``more strongly
entangled'' than $\{\gamma_\mr{sym}(\lambda):\lambda\geq1\}$ in the sense that
the latter can be obtained from the former by deterministic GLOCC but the
reverse is not possible. We leave as an open question whether \emph{all} pure
three-mode Gaussian states can be obtained from $\{\gamma(s,\pi/4):s\geq0\}$
by GLOCC. Since a TMS operation acting on the first mode allows to arbitrarily
reduce the parameter $s\geq0$ (without changing $\theta$) \cite{GECP03} a
positive answer would imply that there is a single (unnormalizable) pure
three-mode state from which all others can be obtained by GLOCC.

Let us finally remark on the entanglement properties associated with the
appearance of a positive determinant (say
$|D_{12}|=(\lambda_3^2+1-\lambda_2^2-\lambda_1^2)/2>0$): First, it means that
$\lambda_1,\lambda_2$ are too small (relative to $\lambda_3$), i.e., there is
too little entanglement available between the modes (12): most (or in the case
of the distributed two-mode squeezed states: all) of the mixedness at these
modes arises from the entanglement with mode 3. Since a two-mode Gaussian
state is necessarily separable if the off-diagonal block of its CM has
non-negative determinant \cite{Sim00}, we see that in that case there is no
residual entanglement between modes (23). As we have seen this strong
concentration of entanglement into one mode cannot be generated by GLOCC. On
the other hand, we have seen that a GLOCC on mode $1$ allows to induce
residual entanglement between the modes (23) (e.g., by generating a
symmetric state) even if their reduced state was separable initially.\\
For the special case $\theta=\pi/4$ and for a suitably chosen Gaussian
operation (essentially, for sufficiently large $x$ and $r$), one can readily
check, using the partial-transpose separability criterion (e.g., in the simple
form for Gaussian $1\times1$ states given in \cite{GDCZ01}), that all three
reduced CMs are entangled.

\section{Conclusions}
We presented an easily computable necessary and sufficient condition for
Gaussian LU--equivalence for an arbitrary number of modes and derived a
standard from for pure three--mode Gaussian states. This showed, in particular
that the entanglement properties of an arbitrary pure three--mode Gaussian
state are completely characterized by three bipartite entanglement measures,
namely the local purities.  This also shows that for pure three-mode Gaussian
states LU--equivalence implies GLU--equivalence.

In order to gain more insight into the relation among the GLU--classes, we
investigated the more general set of GLOCC operations. We provided simple
expressions for GLOCC-transformations between different GLU-classes. For the
pure three-mode states we showed that they can be divided in two classes
(according to whether the sign of the largest determinant $|D_{ij}|$ is
positive or not) such that no GLOCC can transform a state from the second
class to the first. In particular, this shows that the set of symmetric states
(GHZ/W states) does not suffice to generate an arbitrary state via
GLOCC. Among the states unreachable from the symmetric states we identified a
family which, in contrast, allows to prepare all symmetric states.

There are many questions concerning the GLOCC-interconvertibility of pure
multipartite Gaussian states that remain to be addressed: Is there a
``maximally entangled'' family \cite{VSK13} in the sense that \emph{all} other
states can be obtained from it by GLOCC? Is there a majorization relation
governing which states can be GLOCC-transformed into another?  Are there
mutually inaccessible subsets of GLU-classes similar to the W and GHZ classes
for three qubits?  Can the observed restrictions on Gaussian state
transformations be lifted if several copies of the states are considered?  Are
there examples in which general (i.e., non-Gaussian) local unitaries allow the
transformation between two pure Gaussian states that are not in the same
GLU-class or does LU-equivalence of Gaussian states always imply their
GLU-equivalence?  Answers to these questions might lead to a better
understanding of the structure and qualitative features of pure Gaussian
entanglement and be of practical use regarding which states are the most
versatile in terms of state generation.

\begin{acknowledgments}
We acknowledge the kind hospitality of the Pedro Pascual Center for Physics
and the Benasque Quantum Information workshop where much of this work was
done.

GG acknowledges the project MALICIA and its financial support of the Future
and Emerging Technologies (FET) programme within the Seventh Framework
Programme for Research of the European Commission, under FET-Open grant
number: 265522. BK acknowledges financial support of the Austrian Science Fund (FWF): Y535-N16.
\end{acknowledgments}

\appendix
\section{Pure $1\times1\times1$ states: Standard Form}\label{app:A}
We show here that the condition $\gamma J\gamma = J$ implies that any
three--mode CM $\gamma$ is $xp$--blockdiagonal (see Theorem
\ref{th:StdForm}). Let $\gamma$ as given in Eq.~(\ref{eq:1}) denote the
standard form of the CM. In particular, $K=D_{12}$ is diagonal. However,
instead of choosing $O_3$ such that $L=OD_{13}$, for some orthogonal matrix
$O$ and some diagonal matrix $D_{13}$, we chose here without loss of
generality $O_3$ such that $L$ has upper-triangular form \footnote{Note that
  this is always possible for any $L$; as we will see in the case considered,
  both conventions coincide.}. The necessary condition for $\gamma$ to
correspond to a pure state, $\gamma J\gamma = J$, is equivalent to the
following set of equations,
\begin{subequations}
  \begin{align}
    1&=  \lambda_1^2+| D_{12}| + |L|,\\
    1&=\lambda_2^2 + |D_{12}| + |M|,\\
    1&=\lambda_3^2 + | L| + |M|,\\
    \label{eq:2cond4}
    0&=\lambda_1 JD_{12} + \lambda_2 D_{12}J + LJM^T,\\
    \label{eq:2cond5}
    0&=\lambda_1 JL + \lambda_3 LJ + D_{12}JM^T,\\
    \label{eq:2cond6}
    0&=\lambda_2 JM + \lambda_3 MJ + D_{12}^TJL,
  \end{align}
\end{subequations}
Note that $\lambda_i\geq1$ (in particular $\lambda_i \not=0$) for all $i$, must hold for any CM (see e.g. condition \Eqref{iii}). Let us use the notation $x_{ij}=X_{ij}$ for $X\in \{K,L,M\}$. Writing  Eqs.~(\ref{eq:2cond4}-\ref{eq:2cond6}) elementwise
we obtain
\begin{subequations}
\begin{equation}
  \label{eq:AppA1}
  \begin{split}
0&= l_2m_{21},\\
0&=l_{12}m_1-l_1m_{12},\\
0&=\lambda_1 k_2+\lambda_2 k_1-l_{12}m_{21}+l_1m_2,\\
0&=\lambda_1 k_1+\lambda_2 k_2+l_2m_1,
\end{split}
\end{equation}
\begin{equation}
  \label{eq:AppA2}
  \begin{split}
0&=\lambda_1 l_{12}+k_2m_{12},\\
0&=\lambda_3 l_{12}-k_1m_{21},\\
0&=\lambda_1 l_2+\lambda_3 l_1+k_1m_2,\phantom{+m_{12}m_{21}}\\
0&=\lambda_1 l_1+\lambda_3 l_2+k_2m_1,
  \end{split}
\end{equation}
\begin{equation}
  \label{eq:AppA3}
  \begin{split}
0&=\lambda_2 m_{21}-\lambda_3 m_{12},\\
0&=k_1l_2+\lambda_2 m_2+\lambda_3 m_1,\\
0&=k_2l_1+\lambda_2 m_1+\lambda_3 m_2,\phantom{+l_{12}m_{12}}\\
0&=k_2l_{12}+\lambda_2 m_{12}-\lambda_3 m_{21}.
\end{split}
\end{equation}
\end{subequations}
We show that the above equations imply that $l_{12}=m_{12}=m_{21}=0$, i.e.,
also $L$ and $M$ are diagonal. We first discuss the case $l_2\not=0$. Then
the first of \Eqsref{eq:AppA1} implies $m_{21}=0$. Consequently, the second
equation of 
\Eqsref{eq:AppA2} implies $l_{12}=0$ and the second of \Eqsref{eq:AppA2} yields
$m_{21}=0$. If, instead, $l_2=0$, we have that $|L|=0$ and we can
wlog \footnote{By choosing an appropriate orthogonal transformation at mode 3.}
assume $l_{12}=0$. Now consider first $L\not=0$, i.e.,
$l_1\not=0$. Then the second of \Eqsref{eq:AppA1} yields $m_{12}=0$, the
first of \Eqsref{eq:AppA3} implies $m_{21}=0$. If, finally, $L=0$, then
$k_1+k_2=0$ and $m_1+m_2=0$ [Eqs.(\ref{eq:AppA1},\ref{eq:AppA3})], respectively, and then by
the last two of \Eqsref{eq:AppA2} $K=0$ or $M=0$, i.e., mode 1 or
mode 3 factorizes. In either case, both $L$ and $M$ are diagonal and therefore $\gamma$ is $xp$--blockdiagonal.

\section{Proof of Lemma \ref{lemmapure1}}\label{AppB}
Here we present the detail of the proof of Lemma \ref{lemmapure1}. In
particular, we derive the conditions under which $\gamma$ as given in
Eq.~(\ref{GammaDiag}) obeys the necessary condition $\gamma J \gamma=J$.  In
order to increase the readability of the appendix, we restate the equivalent
conditions given in Eq.~(\ref{Eqgammapure}):

\begin{subequations}
\label{AppEqgammapure}
\bea
\lambda_1^2 + |D_{12}|+|D_{13}|=1, \label{Eq1}\\
\lambda_2^2 + |D_{12}|+|D_{23}|=1, \label{Eq2}\\
\lambda_3^2 + |D_{13}|+|D_{23}|=1, \label{Eq3}\\
\lambda_1 D_{12}+\lambda_2 \tilde{D}_{12}+ \tilde{D}_{13}\odot D_{23}=0,\label{Eq4}\\
\lambda_1 D_{13}+\lambda_3 \tilde{D}_{13}+ \tilde{D}_{12}\odot D_{23}=0,\label{Eq5}\\
\lambda_2 D_{23}+\lambda_3 \tilde{D}_{23}+ \tilde{D}_{12}\odot D_{13}=0.\label{Eq6}
\eea
\end{subequations}
As before, $\odot$ denotes the componentwise multiplication (Hadamard
product). Here, we use the index--free notation $D_{12}=\mbox{diag}(a,b)$,
$D_{13}=\mbox{diag}(c,d)$, and $D_{23}=\mbox{diag}(e,f)$, and that
$DJ=J\tilde{D}$, (i.e., $\tilde{D}=-JDJ$)) for any diagonal matrix $D$ and
therefore $DJD=|D|J$. Note that if $D=\mr{diag}(x,y)$, then
$\tilde{D}=\mr{diag}(y,x)$.  In order to solve those equations we distinguish
between the following two cases 
\bi 
\item[i)] at least one of the diagonal matrices, $D_{ij}$ is not invertible and
\item[ii)] none of the determinants vanishes.
\ei

Let us first consider the case (i). Since we do not impose any order on the
$\lambda_i$ we assume without loss of generality that $e=0$. It is then
straight forward to verify that the solution to Eq.~(\ref{AppEqgammapure}) is
given by
\begin{subequations}
\begin{align}
\label{case i}
 \lambda_1&=\sqrt{-1+\lambda_2^2+\lambda_3^2},\\ 
a&=(-1)^{k_1}\sqrt{\lambda_2 (-1+\lambda_2^2)}/\sqrt{\lambda_1},\\ 
 b&=-(-1)^{k_1}\sqrt{-1/\lambda_2+\lambda_2} \sqrt{\lambda_1},\\ 
 c&=(-1)^{k_2}\sqrt{
  \lambda_3 (-1+\lambda_3^2)/\lambda_1},\\ 
d&=-(-1)^{k_2}\sqrt{-1/\lambda_3+\lambda_3} \sqrt{\lambda_1},\\ 
e&=0, \\ 
f&=(-1)^{k_1+k_2}\sqrt{(\lambda_2^2-1)(\lambda_3^2-1)}/\sqrt{(\lambda_2
  \lambda_3)},
\end{align}
\end{subequations}
where $k_1,k_2\in\{0,1\}$. Now, it is easy to see that the four solutions for
the different values of $k_1,k_2$ are GLU--equivalent by choosing
$O=(-1)^{k_1}\one \oplus \one \oplus (-1)^{k_1+k_2}\one$. Thus, we chose
without loss of generality $k_1=k_2=0$.  Given the expressions of the entries
of the diagonal matrices $D_{ij}$ [see \Eqref{Eq_entriesD}] it is straight
forward to check that $a^2= (d_{12}^-)^2$, $b^2= (d_{12}^+)^2$, $c^2=
(d_{13}^-)^2$, $d^2= (d_{13}^+)^2$, $e^2=0= (d_{23}^-)^2$, and $f^2=
(d_{23}^+)^2$. Moreover, it is easy to see that $|D_{ij}|=d_{ij}^+ d_{ij}^-$,
for all three matrices. Thus, the expressions coincide up to a (independent)
global phase for the matrices $D_1,D_2$ (the sign of $D_3$ is thereby
fixed). Let us denote these signs by $k_1,k_2,k_3$ respectively. Clearly,
$d_{12}^+\geq 0$, which implies, since $b\leq 0$, that
$(-1)^{k_1}=-1$. Similarly, it is easy to see that $(-1)^{k_2}=-1$ and
$(-1)^{k_3}=1$ (which has to coincide with $(-1)^{k_1+k_2}$). Thus, the
orthogonal matrix $-\sigma_x\oplus -\sigma_x\oplus \sigma_x$ (corresponding to
a GLU) sorts the entries in the diagonal matrices and applies the right signs
to map $\gamma$ into the form of Eq.~(\ref{GammaDiag}), with the diagonal
entries given in Eq.~(\ref{Eq_entriesD}).

Let us now consider the more involved case (ii). First note that due to
Eq.~(\ref{AppEqgammapure}) the following relations hold
\begin{align}
ab &=|D_{12}|=1/2(1-\lambda_1^2-\lambda_2^2+\lambda_3^2), \\ \nonumber
cd &=|D_{13}|=1/2(1-\lambda_1^2+\lambda_2^2-\lambda_3^2), \\ \nonumber
ef &= |D_{23}|=1/2(1+\lambda_1^2-\lambda_2^2-\lambda_3^2).
\end{align}
Note that two of these determinants are non--positive. More precisely, if
$\lambda_i \geq \lambda_k, \lambda_l$ then $|D_{ik}|\leq 0$ and $|D_{il}|\leq
0$. Let us now define
\begin{align}
x_1&=\frac{1}{4\lambda_2 \lambda_3  |D_{13}|^2},\\ \nonumber
x_2&=\lambda_2^6+\left(-1+\lambda_1^2\right)^2 \lambda_3^2-2 \left(1+\lambda_1^2\right) \lambda_3^4+\lambda_3^6 -\\ \nonumber
&{}-\lambda_2^4 \left(2+2 \lambda_1^2+\lambda_3^2\right)+\\ \nonumber
&{}+\lambda_2^2 \left(\left(-1+\lambda_1^2\right)^2+4 \left(1+\lambda_1^2\right) \lambda_3^2-\lambda_3^4\right).
\end{align}
The solution to Eq.~(\ref{AppEqgammapure}) is then given by
\begin{subequations}
\begin{align}
\label{eq:f}
f&=2\frac{y_1 x}{y_2+y_3 x^2},\\
\label{eq:e}
e&=\frac{|D_{23}|}{f},\\ 
\label{eq:d}
d&= \frac{\sqrt{\lambda_1}\sqrt{-|D_{12}|}}{\sqrt{ex+\lambda_2x^2}},\\
\label{eq:c}
c&=\frac{|D_{13}|}{d}, \\ 
\label{eq:b}
b&= xd, \\
\label{eq:a}
a&= \frac{|D_{12}|}{b}.
\end{align}
\end{subequations}
Here we have used
\begin{align}
\label{eq:x}
x&=(-1)^k\left(\frac{1}{\sqrt{2}}\sqrt{x_1\left(x_2+(-1)^l\sqrt{
        w}\right)}\right),\\
 \label{Eq:w}
w&=(-1+\lambda_1-\lambda_2-\lambda_3) (1+\lambda_1-\lambda_2-\lambda_3) \times
\\ \nonumber
&(-1+\lambda_1+\lambda_2-\lambda_3) (1+\lambda_1+\lambda_2-\lambda_3) \times
\\ \nonumber
&(-1+\lambda_1-\lambda_2+\lambda_3) (1+\lambda_1-\lambda_2+\lambda_3) \times
\\ \nonumber
&(-1+\lambda_1+\lambda_2+\lambda_3) (1+\lambda_1+\lambda_2+\lambda_3)\left(\lambda_2^2-\lambda_3^2\right)^2,\\
y_1&=(\lambda_3^2-\lambda_2^2)|D_{23}|,\\
y_2&=-2\lambda_3|D_{12}|,\\
y_3&=2\lambda_2|D_{13}|,
\end{align}
and $k,l\in \{0,1\}$.

Note that the denominator of $x_1$ is non--vanishing since $D_{13}$ is
invertible. Note further that the denominator of $f$ is vanishing (for
$\lambda_i\geq 1$) iff either (a) $\lambda_1=\sqrt{1-\lambda_2^2+\lambda_3^2}$
or (b) $\lambda_j=\sqrt{1-\lambda_k^2+\lambda_1^2}$, for $j\neq k$, or (c)
$\lambda_i=\sqrt{-1+\lambda_j^2+\lambda_k^2}$ or (d)
$\lambda_2=\lambda_3$. The cases (a-c) cannot occur here, since in those cases
one of the determinants $D_{ij}$ vanishes. Let us now first consider the case
$\lambda_2\neq \lambda_3$ (for case (d), $\lambda_2= \lambda_3$ a similar
argument applies).

Note that $\gamma$ is real only if $w\geq 0$. Let $\lambda_i\geq
\lambda_k,\lambda_j$, for mutually different values of $i,j,k\in \{1,2,3\}$
denote the largest value, then it can be easily seen that $w\geq 0$ iff either
$\lambda_i\leq \lambda_j+\lambda_k-1$ or $\lambda_i\geq
\lambda_j+\lambda_k+1$. As shown in the main text, the second choice is
excluded due to the positivity of $\gamma$.

Note further that all values of $k,l$ lead to a solution. Those equalities
have been derived as follows.  First we use the conditions
$\lambda_1^2+|D_{12}|+|D_{13}|=1, \lambda_2^2+|D_{12}|+|D_{23}|=1,
\lambda_3^2+|D_{13}|+|D_{23}|=1$ to compute $a, c, e$ as functions of the other
parameters. As can be easily seen, the conditions given in
Eq.~(\ref{AppEqgammapure}) imply that $b c (f \lambda_2 + e \lambda_3) - a d
(e \lambda_2 + f \lambda_3)=0$, which implies that $b=xd$, where $x$ is a
function which depends only on $f$ and $\lambda_i$. Using then that
$-de-a\lambda_1-b\lambda_2=0$ we compute $d$ as a function of $f$, $x$, and
$\lambda_i$. Next we compute $c$ as a function of $f$, $x$, and $\lambda_i$ by
using that $-be-c\lambda_1-d\lambda_3=0$. Thus, we have all variables as
functions of $f$, $x$, and $\lambda_i$. Using then the condition $(\gamma
J\gamma- J)_{2,1}=0$ we derive $f=y_1 x/(y_2+y_3x^2)$. The equation $(\gamma
J\gamma- J)_{1,4}=0$ allows us then to compute $x$ as given above. Note that
we obtain two solutions for $d$, namely $\pm d$ for $d$ given in
Eq.~(\ref{eq:d}). However, changing the sign of $d$ amounts to changing the
signs of $c,a,b$ (cf. Eqs.~(\ref{eq:c}-\ref{eq:a})) and corresponds therefore
to the GLU $O_1=-\one$, $O_2=O_3=\one$.

It is tedious, but straight forward to show that all four solutions, $k,l\in
\{0,1\}$ are GLU--equivalent to the one with $k=l=0$.

Similarly to the case (i) it can now be shown that the expressions we derived
for the entries of the diagonal matrices coincide with the ones given in
Eq.~(\ref{Eq_entriesD}). However, here we have that $a=d_{12}^+$, etc. For
$\lambda_2=\lambda_3$ a similar argument can be used to arrive at the same
conclusion, which completes the proof.

As shown in the following section appendix, the necessary condition
that $\gamma\geq 0$ is equivalent to the condition given in
Eq.(\ref{CondPos}).

Note that this implies that given the three local purities $\lambda_i$ (or
equivalently the bipartite entanglement shared in the three splittings
$i|jk$), the state is uniquely determined. The reason for that is that
$\lambda_1=\sqrt{-1+\lambda_2^2+\lambda_3^2}$ iff $e=0$ [also in case (ii)]
and therefore, knowing the parameters $\lambda_i$ implies that we also know
which of the two cases the state belongs to. Thus, an arbitrary state is
uniquely determined (up to GLUs) by the bipartite entanglement.

\section{Positivity of $\gamma(\lambda_1,\lambda_2,\lambda_3)$}\label{AppC}
To see that the conditions
\begin{equation}
  \label{eq:AppC-cond1}
\lambda_i+\lambda_j\geq\lambda_k+1\,\forall (ijk)
\end{equation}
(cf. \Eqref{CondPos})
imply positivity of the CM $\gamma=\gamma(\lambda_1,\lambda_2,\lambda_3)$
we proceed as follows: $\gamma$ is by construction $xp$-blockdiagonal and
since it
has been constructed to satisfy the purity condition $\gamma J\gamma=J$, it
follows that $\gamma_p=\gamma_x^{-1}$, hence positivity of $\gamma_x$ implies
positivity of $\gamma$. Using the Schur complement \cite{HJ87} positivity of
$\gamma_x$ is, as $\lambda_3>0$, equivalent to positivity of the $2\times2$ matrix $Y$
\begin{equation}
  \label{eq:11}
  Y = \left( \begin{array}{cc} \lambda_1&d_{12}^+\\ d_{12}^+ & \lambda_2
\end{array} \right)-\left( \begin{array}{c} d_{13}^+\\ d_{23}^+
\end{array} \right)\frac{1}{\lambda_3}\left( \begin{array}{cc} d_{13}^+& d_{23}^+
\end{array} \right),
\end{equation}
which is equivalent to the two conditions
\begin{align}\label{eq:AppC-cond2a}
  \tr Y&\geq0,\\
\label{eq:AppC-cond2b}
\det Y&\geq0.
\end{align}
The trace is found to be
\[
\frac{\lambda_1 + \lambda_2}{8 \lambda_1 \lambda_2 \lambda_3^2}
\left(K_1-1 - \sqrt{w_1}\right),
\]
where we have introduced
\begin{align}
  \label{eq:12}
K_1&=-\sum_i \lambda_i^4+2\sum_i(\lambda_j^2\lambda_k^2+\lambda_i^2),\\
w_1&=(-1+\lambda_1-\lambda_2-\lambda_3)
(1+\lambda_1-\lambda_2-\lambda_3)\times\\
&\nonumber{}\times(-1+\lambda_1+\lambda_2-\lambda_3) (1+\lambda_1+\lambda_2-\lambda_3) \times \\
\nonumber &{}\times (-1+\lambda_1-\lambda_2+\lambda_3)
(1+\lambda_1-\lambda_2+\lambda_3)\times\\
&\nonumber{}\times (-1+\lambda_1+\lambda_2+\lambda_3)
(1+\lambda_1+\lambda_2+\lambda_3).
\end{align}
It follows directly from \Eqref{eq:AppC-cond1} that $w_1\geq0$. It is
tedious, but straight forward to show that
\begin{align}
\det Y&= \frac{\tr Y}{\lambda_1+\lambda_2}.
\end{align}
Thus we see that both conditions
Eqs.~(\ref{eq:AppC-cond2a},\ref{eq:AppC-cond2b}) hold and therefore
$\gamma\geq0$ if
\begin{equation}
  \label{eq:PosCond}
  K_1-1\geq\sqrt{w_1}.
\end{equation}
To see that $K_1-1\geq0$ we write it as a sum of positive terms (using that the conditions given in Eq.~(\ref{CondPos}) are satisfied):
\begin{align*}
  K_1-1&=\frac{1}{4}(\left[(\lambda_3-1)^2-(\lambda_2-\lambda_1)^2\right]\left[(\lambda_1+\lambda_2)^2-(\lambda_3+1)^2\right]+\\
&{}+\left[\lambda_3^2-(\lambda_1-\lambda_2-1)^2\right]\left[(\lambda_1+\lambda_2-1)^2-\lambda_3^2\right]+\\
&{}+\left[\lambda_3^2-(\lambda_1-\lambda_2+1)^2\right]\left[(\lambda_1+\lambda_2-1)^2-\lambda_3^2\right]+\\
&{}+\left[(\lambda_3-1)^2-(\lambda_1-\lambda_2)^2\right]\left[(\lambda_1+\lambda_2)^2-(\lambda_3-1)^2\right])\\
&{}+\sum_i\lambda_i^2\left(
  -\lambda_i+\lambda_j+\lambda_k\right)+\sum_i\lambda_i(2\lambda_i-1)+2\Pi_l\lambda_l,
\end{align*}
where $(jk)$ in $\sum_i(\lambda_j+\lambda_k)^2$ refer in each term to the two
indices distinct from $i$. Now the the remaining condition
$K_1-1\geq\sqrt{w_1}$ can be checked for the
squares of both sides and we find it trivially satisfied:
\[
(K_1-1)^2-w_1=64\lambda_1^2\lambda_2^2\lambda_3^2\geq0.
\]
Therefore both $\det D\geq0$ and $\tr D\geq0$ and consequently
$\gamma(\lambda_1,\lambda_2,\lambda_3)\geq0$ whenever the $\lambda$'s satisfy Eq.
\eqref{CondPos}.

%

\end{document}